\documentclass[a4paper,10pt]{article}
\usepackage{braket} % Allows for Bra-Ket notation automatically.
\usepackage{bm} % Allows for bold in math. \bm{ $math$ }
\usepackage{cancel} % Allows for canceling terms graphically
\usepackage{enumerate}
\usepackage{amssymb, amsmath, amsthm}
\usepackage{multirow}
\usepackage{mathrsfs} % For \mathscr (script) font
\usepackage{breqn} % Break equation, dmath
\usepackage{fancyhdr} % Required for custom headers
\usepackage{lastpage} % Required to determine the last page for the footer
\usepackage{extramarks} % Required for headers and footers
\usepackage{changepage}
\usepackage{graphicx}
\usepackage{empheq}
\usepackage{caption,subcaption}
\usepackage{xfrac} %xfrac allows for different looking fractions. Use
\usepackage{graphics,setspace} % Used for Mathematica Output
\usepackage{pgfplots}
\usepackage{relsize}
\usepackage{floatrow}
\usepackage{bbm} %for \mathbbm
\usepackage{nth}
\usepackage{color}
\usepackage{listings}
\usepackage{tikz}
\usepackage{hyperref}
\usetikzlibrary{spy}

\newfloatcommand{capbtabbox}{table}[][\FBwidth]

\DeclareMathOperator{\RE}{Re} %Re - Real part
\DeclareMathOperator{\IM}{Im} %Re - Real part
\DeclareMathOperator{\Tr}{Tr} %Trace
\AtBeginDocument{\renewcommand{\d}{\textrm{d}}}

\newcommand{\cc}{^*}

\newcommand{\pd}[2]{\frac{\partial #1}{\partial #2}}
\newcommand{\pdd}[3][]{\frac{\partial^{#1} #2}{\partial #3^{#1}}}

\newcommand{\D}[2][]{\frac{\d^{#1}}{\d #2^{#1}}}
\newcommand{\abs}[1]{\left| #1 \right|}

\newcommand{\tr}{\text{tr}} %\trace
\newcommand{\R}{\mathbb{R}} %Real.. finally 
\newcommand{\N}{\mathbb{N}} %Natural numbers
\newcommand{\C}{\mathbb{C}} %Complex!
\newcommand{\cP}{\mathcal{P}} %Polynomial space.
\newcommand{\cS}{\mathcal{S}} %Symmetrize
\newcommand{\cU}{\mathcal{U}} 
 
\newcommand{\cQ}{\mathcal{Q}} 
\newcommand{\cE}{\mathcal{E}} %
\newcommand{\cL}{\mathcal{L}} %
\newcommand{\lt}{\left}
\newcommand{\rt}{\right}
\newcommand{\Langle}{\lt\langle}
\newcommand{\Rangle}{\rt\rangle}

\newcommand{\eps}{\epsilon}

\newcommand{\lam}{\lambda}

\newcommand{\0}[1]{\mathcal{O}\lt(#1\rt)}
\newcommand{\bigo}[1]{\mathcal{O}} %Conforming with Randy

\let\oldchi\chi
\renewcommand{\chi}{\raisebox{1pt}{$\oldchi$}}

\newcommand{\vs}{\textit{vs.}~}

\newenvironment{definition}[1][Definition.]{\begin{trivlist}
\item[\hskip \labelsep {\bfseries #1}]}{\end{trivlist}}

\newtheorem{theorem}{Theorem}
\newtheorem{lemma}[theorem]{Lemma}

\theoremstyle{definition}
\newtheorem{remark}[theorem]{Remark}

\newtheorem*{example*}{Example}

\theoremstyle{plain}
\newtheorem{assumption}[theorem]{Assumption}

\usepackage[utf8]{inputenc}
\usepackage{epstopdf}
\usepackage{enumitem}
\usepackage[title]{appendix}
\usepackage{cancel}
\usetikzlibrary{arrows}
\usetikzlibrary{decorations.markings}
\usepackage{mathtools}
\newcommand{\te}[1]{{\leavevmode\color{black}#1}}

\title{Real Lax spectrum implies spectral stability}
\author{Jeremy Upsal and Bernard Deconinck
  \\
Department of Applied Mathematics,\\
University of Washington,\\
Seattle, WA 98195, USA}

\begin{document}

\maketitle
\begin{abstract}
We consider the dynamical stability of periodic and quasiperiodic stationary
solutions of integrable equations with 2$\times$2 Lax pairs. We construct the
eigenfunctions and hence the Floquet discriminant for such Lax pairs. The
boundedness of the eigenfunctions determines the \emph{Lax spectrum.} We use
the squared eigenfunction connection between the Lax spectrum and the stability
spectrum to show that the subset of the real line that gives rise to stable
eigenvalues is contained in the Lax spectrum.  For non-self-adjoint members of
the AKNS hierarchy admitting a common reduction, the real line is always part
of the Lax spectrum and it maps to stable eigenvalues of the stability problem.
We demonstrate our methods work for a variety of examples, both in and not in
the AKNS hierarchy.
\end{abstract}

\section{Introduction}
A surprisingly large number equations of physical significance are integrable
and possess a Lax pair \cite{AKNS74, SolitonsAndIST}. An important feature of
equations with a Lax pair is the Lax spectrum: the set of all Lax parameter
values for which the solution of the Lax pair is bounded. For our purposes,
this is important for determining the stability of solutions of a given
integrable equation \cite{ bottman2009kdv, bottman2011elliptic,
calini2011squared, deconinck2010orbital, SGstability, mkdvOrbitalStability,
deconinck2017stability, upsalNLS, forest1982spectral, leeThesis,
kdvFGOrbitalStability}.  Until recently, the Lax spectrum has only been
determined explicitly for decaying potentials on the whole line or for
self-adjoint problems with periodic coefficients. For such problems, the
Floquet discriminant \cite{ ablowitz1996computational, calini2011squared,
MR1123280, forest1982spectral, ivey2008spectral, leeThesis} is a useful tool
for numerically computing and giving a qualitative description of the Lax
spectrum, but it is not used generally to get an explicit description of the
Lax spectrum. A full description of the Lax spectrum can allow one to prove the
stability of solutions to integrable equations with respect to certain classes
of perturbations \cite{upsalNLS}.

\te{
In Section 2 we define the type of stability we study for general nonlinear
evolution equations. In Section 3 we restrict to a special class of evolution
equations, those in the AKNS hierarchy. We construct a function defined by an integral,
whose zero level set determines the Lax spectrum for members of the AKNS
hierarchy. Using this, we show that the Lax spectrum for members of the AKNS
hierarchy admitting a common reduction contains a subset of the real line that
maps to stable elements of the stability problem under the
\emph{squared-eigenfunction connection}. For non-self-adjoint members of the
AKNS hierarchy with this reduction, this subset is the whole real line. In
Section 4, we first show how these results apply to
some example equations in the AKNS hierarchy. We finish Section 4 by showing
how the results in Section 3 also apply to an equation that is not in the AKNS
hierarchy. This leads to Section 5 in which we generalize the results of
Section 3 to integrable equations possessing a $2\times 2$ Lax pair that
are not in the AKNS hierarchy. In Section 6, we apply the results of Section 5
to a number of examples. At the end of Section 6 we demonstrate how similar
results are obtained for equations that do not fit into the framework of
Section 5, suggesting that the work here may be more general than what we
present here. For the interested reader, we construct the Floquet discriminant for these problems in the
Appendix.
}

\section{Stability setup}
We consider the nonlinear evolution equation
\begin{align}
    u_t &= \mathcal{N}(u, u_x, \ldots, u_{Nx}),\label{eqn:generalNLEE}
\end{align}
where $u(x,t)$ is a real- or complex-valued function (possibly vector valued)
and $\mathcal{N}$ is a nonlinear function of $u$ and $N$ of its spatial
derivatives. We assume that \eqref{eqn:generalNLEE} is written in a frame in
which there exists a nontrivial stationary solution, $u(x,t) = \bar u(x)$. In
order to study the stability of the solution $\bar u$, we linearize about it by
letting $u(x,t) = \bar u(x) + \eps v(x,t) + \0{\eps^2}$. Truncating at
$\0{\eps}$ yields
\begin{align}
  v_t = \mathcal{L}(\bar u, \bar u_x, \ldots, \bar u_{Nx})v,
\end{align}
where $\mathcal{L}$ is a linear function. Since $\mathcal{L}(\bar u, \bar u_x,
\ldots, \bar u_{Nx})$ is independent of $t$ we may separate variables with
\begin{align}
  v(x,t) = \hat v(x) e^{\lam t} \label{eqn:stabSeparate},
\end{align}
yielding the spectral problem for $\lam$:
\begin{align}
  \lam \hat v = \mathcal{L}(\bar u, \bar u_x,\ldots,\bar u_{Nx})\hat
  v.\label{eqn:stab}
\end{align}
\begin{definition}
    The \emph{stability spectrum} is the set
    \begin{align}
      \sigma_{\mathcal{L}} = \{ \lam \in \C : \hat v\in S_{\cL}\},
    \end{align}
    where $S_{\cL}$ is a function space \te{that depends on the perturbations
	of interest.} In essence, we choose $S_{\cL}$ to be the least restrictive space
	that works in our framework.
\end{definition}

We are interested in equations \eqref{eqn:generalNLEE} that are Hamiltonian.
The stability spectrum for such equations possesses a quadrafold symmetry: if
$\lam \in \sigma_{\cL}$, then $-\lam, \lam\cc, -\lam\cc \in \sigma_{\cL}$ (here
and throughout~$^*$~represents the complex conjugate)
\cite{haragus_kapitula-SpectraPeriodicWaves}. There are many different types of
stability. In this paper, whenever we refer to \emph{stability} we are using
the following definition of \emph{spectral stability}.
\begin{definition}
    A solution $u(x,t) = \bar u(x)$ of \eqref{eqn:generalNLEE} is \emph{stable}
    if and only the stability spectrum of the corresponding operator $\cL =
    \cL(\bar u, \bar u_x,\ldots, \bar u_{Nx})$ is a subset of the imaginary
    axis, $\sigma_{\cL}\subset i\R$. 
\end{definition}

Determining the stability spectrum is usually not a straightforward task.
Since we focus on \emph{integrable equations}, significant progress can be
made. We start with a special class of integrable equations known to be in the
\emph{AKNS hierarchy} \cite{AKNS74}.

\section{The AKNS hierarchy}\label{sec:AKNS}

\subsection{The Lax pair, the Lax spectrum, and the Squared-eigenfunction
connection}
We define an integrable equation to be an evolution equation of the form
\eqref{eqn:generalNLEE} that possesses a \emph{Lax pair}. A Lax pair is a
pair of ODEs of the form $\phi_x = X\phi,~\phi_t = T\phi$ where $X$ and $T$ are
operators acting on a function $\phi$. The integrable equation
\eqref{eqn:generalNLEE} is obtained by requiring that $\partial_t \phi_x =
\partial_x \phi_t$ holds. 

The AKNS hierarchy \cite{AKNS74} is a special class of integrable equations
containing a variety of physically important nonlinear evolution equations. The
Lax pair for members of the AKNS hierarchy is 
\begin{subequations}
\label{eqn:AKNSLaxPair}
\begin{align}
\label{eqn:AKNSLaxPairx}
\phi_x(x,t;\zeta) &= \begin{pmatrix*}
    -i\zeta & q(x,t) \\
    r(x,t) & i\zeta
  \end{pmatrix*} \phi(x,t;\zeta) = X\phi,\\
\label{eqn:AKNSLaxPairt}
\phi_t(x,t;\zeta) &= \begin{pmatrix*}[r]
    A(x,t;\zeta) & B(x,t;\zeta) \\
    C(x,t;\zeta) & -A(x,t;\zeta)
  \end{pmatrix*}\phi(x,t;\zeta) = T\phi.
\end{align}
\end{subequations}
Here $\zeta \in \C$ is called the \emph{Lax parameter}, assumed to be
independent of $x$ and $t$, and $r,~q,~A,~B,$ and $C$ are complex-valued
functions chosen such that the compatibility of mixed derivatives, $\partial_t
\phi_x = \partial_x \phi_t$ holds if and only if \eqref{eqn:generalNLEE} holds.
The compatibility condition defines the evolution equations~\cite{AKNS74}
\begin{subequations}\label{eqn:AKNStimeEvolution}
\begin{align}
  q_t &= B_x + 2 i \zeta B + 2 A q, \label{eqn:qEvolution}\\
  r_t &= C_x - 2 i \zeta C - 2 A r, \label{eqn:rEvolution}
\end{align}
\end{subequations}
as well as the condition
\begin{align}
  A_x &= qC - rB \label{eqn:AeqnAKNS}.
\end{align}
We are interested in studying the stability of stationary and periodic or
quasiperiodic solutions of members of the AKNS hierarchy with
\begin{align}
  r = \kappa q\cc,
  \label{eqn:involutions}
\end{align}
where $\kappa = \pm 1$. We make the following assumption.

\begin{assumption}\label{assumption:stationary}
    The functions $q$ and $r$ are related by \eqref{eqn:involutions}, are
    $t$-independent, and have the form
    \begin{align}
      q(x) = e^{i\theta(x)}Q(x), \quad \text{and}\quad r(x) = \kappa
      e^{-i\theta(x)}Q(x),\label{eqn:AKNSInvolution}
    \end{align}
    where $Q$ and $\theta$ are real-valued functions and $Q\geq 0$ is
    $P$-periodic.
\end{assumption}

\noindent With this assumption, \eqref{eqn:AKNStimeEvolution} gives
\begin{subequations}\label{eqn:AKNScompatibility}
\begin{align}
  B_x &= -2 i \zeta B - 2 A q, \label{eqn:BeqnAKNS}\\
  C_x &=  2 i \zeta C + 2 A r. \label{eqn:CeqnAKNS}
\end{align}
\end{subequations}

\begin{remark}\label{remark:recursionOp}
Many evolution equations are found by assuming $A,~B,$ and $C$ can be written as
a $\zeta$-power series \cite{SolitonsAndIST}. In such cases, a recursion
operator is used to find $A,~B,$ and $C$ \cite{SolitonsAndIST} and \cite[Chapter
2]{gerdjikov2008integrable}. Using the recursion operator, $A$ is a function of
products of $q_{nx}^j r_{nx}^j,$ where $n,~j \in \N$.  Similarly, $B$ is a
function of products of the form $q_{nx}^{j+1} r_{mx}^j$ and $C$ is a function
of products of the form $q_{nx}^j r_{mx}^{j+1}$. Therefore, when
Assumption~\ref{assumption:stationary} holds, $A,~B,$ and $C$ are
$t$-independent, bounded for $x\in \R$, including $\infty$, and $A$ is
$P$-periodic. In what follows we assume each of these statements are true.
\end{remark}

\begin{definition}
  The \emph{Lax spectrum} is the set
    \begin{align}
      \sigma_L = \{\zeta \in \C: \phi \in S_{L}\}
      \label{eqn:LaxSpectrumDef},
      \end{align}
      where $S_L$ is a \te{function space that depends on the perturbations of interest.}
It turns out that $S_L$ is related to the
      space $S_{\cL}$. 
\end{definition}

Of course, $\sigma_L$ depends on which norm is used. Although finding the Lax
spectrum is an interesting and important problem in its own right, we are
primarily interested in using the Lax spectrum to determine the stability of
stationary solutions to integrable equations. The connection between the Lax
spectrum and stability is through the so-called \emph{squared-eigenfunction
connection}. The squared-eigenfunction connection (sometimes referred to as
quadratic eigenfunctions) gives a connection involving quadratic combinations
of the eigenfunctions of \eqref{eqn:AKNSLaxPair} between eigenfunctions of
the stability problem \eqref{eqn:stab} and eigenfunctions of the Lax problem
\eqref{eqn:AKNSLaxPair} \cite{AKNS74,gerdjikov2008integrable}. In order to
connect the Lax spectrum with the stability spectrum, we make the following
observation (originally used in \cite{bottman2009kdv}). When $A,~B,$ and $C$
are $t$-independent (see Remark~\ref{remark:recursionOp})
\eqref{eqn:AKNSLaxPairt} may be solved by separation of variables. Equating
\begin{align}
  \phi(x,t) = e^{\Omega t}\varphi(x),\label{eqn:phiSepVar}
\end{align}
where $\Omega$ is complex valued and $\varphi(x)$ is a complex
vector-valued function, \eqref{eqn:AKNSLaxPairt} becomes a $2\times 2$
eigenvalue equation for $\Omega$:
\begin{align}
  \Omega\varphi = T\varphi.
\end{align}
Using the expression \eqref{eqn:AKNSLaxPairt} for $T$,
\begin{align}
  \Omega^2 = A^2 + BC.\label{eqn:Omega}
\end{align}
The following lemma establishes that $\Omega$ is not a function of $x$, in
addition to the obvious fact that it is not a function of $t$. 

\begin{lemma}\label{lem:OmegaIndep}
  Under Assumption~\ref{assumption:stationary}, $\Omega$ is independent of $x$
  and $t$.
\end{lemma}
\begin{proof}
  Independence of $t$ is immediate since $A,~B$ and $C$ are independent of $t$ for
  stationary solutions (Remark~\ref{remark:recursionOp}).  Multiplying
  \eqref{eqn:BeqnAKNS} by $C$ and \eqref{eqn:CeqnAKNS} by $B$ and adding the
  resulting equations yields
  \begin{align}
    0 &= \partial_x(BC) + 2A (q C - r B ).
  \end{align}
  Using \eqref{eqn:AeqnAKNS},
  \begin{align}
    0 &= \partial_x(BC + A^2) = \partial_x(\Omega^2),
  \end{align}
  so $\Omega^2$ is independent of $x$.
\end{proof}

Thus $\Omega = \Omega(\zeta)$ is a function only of $\zeta$ and the solution
parameters. The connection between the Lax spectrum and the stability spectrum
is through $\Omega$.  Comparing the exponential component of quadratic
combinations of the two components of $\phi$ with $v$ from
\eqref{eqn:stabSeparate} yields 
\begin{align}
  \lam = 2\Omega(\zeta) \label{eqn:SQEF}.
\end{align}
The squared-eigenfunction connection is known to be complete on the whole line
\cite{gerdjikov2008integrable}, but this has not yet been shown in other
settings. We have been able to show it is complete in every example that we
have studied in depth \cite{bottman2009kdv, bottman2011elliptic,
deconinck2010orbital, mkdvOrbitalStability, upsalNLS, kdvFGOrbitalStability}.
In such cases, \eqref{eqn:SQEF} gives the entire stability spectrum if the set
$\sigma_L$ is known. In other words, the map
$\Omega:\sigma_{L}\mapsto\sigma_{\cL}$ is surjective. This gives us a
connection between the function spaces defining the stability and the Lax
spectrum: both are defined by the spatial boundedness of the eigenfunctions in
question. The next step for studying stability is to find the Lax spectrum.
Before doing so, we provide an example of how the steps in this section work
for a well-known equation.

\begin{example*}[The nonlinear Schr\"{o}dinger (NLS) equation]
$ $\newline
The nonlinear Schr\"{o}dinger (NLS) equation is
  \begin{align}
    i \Psi_t + \frac12 \Psi_{xx} - \kappa\Psi\abs{\Psi}^2 &= 0,\label{eqn:NLS}
  \end{align}
  where $\Psi(x,t)$ is a complex-valued function and $\kappa = -1$ and $\kappa =
  1$ correspond to the focusing and defocusing equations respectively. The Lax pair
  for the NLS equation \cite{zakharov_shabat} is given by
  \eqref{eqn:AKNSLaxPair} with
  \begin{subequations}
    \begin{align}
     & q = \Psi,& & r = \kappa \Psi\cc,&
     &A = -i\zeta^2 - i\kappa \abs{\Psi}^2/2, &
     &B = \zeta \Psi + i\Psi_x/2,&
     &C = \zeta \kappa \Psi\cc - i\kappa \Psi\cc_x/2.&
    \end{align}
  \end{subequations}
  Note that $A,~B,$ and $C$ satisfy the properties in
  Remark~\ref{remark:recursionOp}.  Equating $\Psi(x,t) = e^{-i\omega
  t}\psi(x,t)$, where $\omega \in \R$ is constant, we obtain the NLS equation in
  a frame rotating with constant phase speed $\omega$,
  \begin{align}
    i \psi_t + \omega \psi + \frac12 \psi_{xx} - \kappa \psi\abs{\psi}^2
    &=0.\label{eqn:NLSTraveling}
  \end{align}
  Equation \eqref{eqn:NLSTraveling} can be obtained from the compatibility
  of the new $t$-equation,
  \begin{align}
    \phi_t &= \begin{pmatrix} -i\zeta^2 -i\kappa \abs{\psi}^2/2 +  i \omega/2 &
      \zeta\psi + i\psi_x/2 \\
      \zeta \kappa \psi\cc - i \kappa \psi_x\cc/2 &
      i\zeta^2 + i\kappa  \abs{\psi}^2/2 - i \omega/2
    \end{pmatrix}\phi,
    \label{eqn:NLSnewTeqn}
  \end{align}
  and the $x$ equation \eqref{eqn:AKNSLaxPairx}, which is unchanged. 

  The (quasi)periodic stationary solutions of \eqref{eqn:NLSTraveling} are called
  the \emph{elliptic solutions}. The stability of the elliptic solutions of the
  defocusing and focusing NLS equation was studied in \cite{bottman2011elliptic}
  and \cite{deconinck2017stability,upsalNLS} respectively. To do so we linearize
  \eqref{eqn:NLSTraveling} about a stationary solution $\tilde\psi(x)$ by
  letting $\psi(x,t) = \tilde \psi(x) + \eps u(x,t)+{\cal O}(\eps^2)$. This
  results in
  \begin{align}
    U_t = \begin{pmatrix} u \\[0.5em] \kappa u\cc\end{pmatrix}_t &= \begin{pmatrix}
      \frac i2 \partial_x^2 - 2i\kappa |\tilde\psi|^2 + i\omega & -i
      \tilde \psi^2 \\[0.5em]
      i\tilde\psi^{\ast 2} & -\frac i2 \partial_x^2 + 2i\kappa |\tilde
      \psi|^2 -i\omega
    \end{pmatrix}
    \begin{pmatrix}
      u \\[0.5em] \kappa u\cc
    \end{pmatrix}
    = \cL_{\text{NLS}}U.
  \end{align}
  Since $\cL_{\text{NLS}}$ does not depend explicitly on $t$, we separate variables with
  \begin{align}
    U(x,t) = \begin{pmatrix} u(x,t) \\ \kappa u\cc(x,t)\end{pmatrix}= e^{\lam
      t}\begin{pmatrix}v(x)\\\kappa v\cc(x)\end{pmatrix} = e^{\lam t}V(x),
  \end{align}
  resulting in the spectral problem
  \begin{align}
    \lam V = \cL_{\text{NLS}}V. \label{eqn:spectral}
  \end{align}
  The squared-eigenfunction connection for the NLS equation
  \cite{bottman2011elliptic} gives
  \begin{align}
    U(x,t) &= \begin{pmatrix} \phi_1 ^2\\[0.5em]
    \phi_2^2\end{pmatrix},\label{eqn:NLSSQEF}
  \end{align}
  where $\phi = (\phi_1,~\phi_2)^\intercal$ is an eigenfunction of
  \eqref{eqn:NLSnewTeqn}. Using \eqref{eqn:phiSepVar}, the eigenfunctions of
  $\cL_{\text{NLS}}$ are given by
  \begin{align}
    U(x,t) &= e^{\lam t}V(x) = \begin{pmatrix}\phi_1^2
    \\[0.5em]\phi_2^2\end{pmatrix} = e^{2\Omega t}
    \begin{pmatrix}\varphi_1^2(x)\\[0.5em]
    \varphi_2^2(x)\end{pmatrix}, \label{eqn:NLSSQEFConnection}
  \end{align}
  hence $\lam = 2\Omega(\zeta)$ \eqref{eqn:SQEF}. The map $\Omega:\sigma_L\to
  \sigma_{\cL_{\text{NLS}}}$ is shown to be surjective in
  \cite{bottman2011elliptic, upsalNLS}.

\end{example*}
\subsection{Finding the Lax spectrum}\label{sec:FindingAKNSLaxSpectrum}
We begin by introducing the isospectral transformation
\begin{align}
  \Phi(x,t) &= \begin{pmatrix} \Phi_1\\ \Phi_2\end{pmatrix} = \begin{pmatrix}
    e^{-i\theta/2}\phi_1 \\[0.5em]
    e^{i\theta/2}\phi_2\end{pmatrix},\label{eqn:Isospectral}
\end{align}
by which the Lax pair \eqref{eqn:AKNSLaxPair} becomes
\begin{subequations}
\begin{align}
  \Phi_x &= \begin{pmatrix} \alpha & Q(x)\\ \kappa Q(x) & -\alpha
  \end{pmatrix}\Phi,\label{eqn:AKNSLaxPairPeriodicx}\\
  \Phi_t &= \begin{pmatrix} A & e^{-i\theta}
    B \\ e^{i\theta} C & -A\end{pmatrix}\Phi = \begin{pmatrix} A & \hat B \\
    \hat C & -A\end{pmatrix}\Phi,\label{eqn:AKNSLaxPairPeriodict}
\end{align}
\label{eqn:AKNSLaxPairPeriodic}
\end{subequations}
where
\begin{align}
  &\alpha = -i\zeta - i\theta_x/2,& &\hat{B} = e^{-i\theta }B, & &\hat C =
  e^{i\theta}C.\label{eqn:AKNSHats}
\end{align}
This form is helpful since, by Remark~\ref{remark:recursionOp}, $\hat B$ and
$\hat C$ are $P$-periodic along with $A$. The compatibility conditions
\eqref{eqn:AeqnAKNS} and \eqref{eqn:AKNScompatibility} become
\begin{align}
  &A_x = Q \hat C - \kappa Q \hat B, & &\hat B_x = 2(\alpha \hat B - A Q), & &
  \hat C_x = -2(\alpha \hat C - \kappa A  Q).&\label{eqn:AKNScompatibilityHat}
\end{align}

To find $\sigma_L$, we find the eigenfunctions of \eqref{eqn:AKNSLaxPair} using
a technique first used in \cite{bottman2011elliptic}. We note that the
eigenfunctions can be found using other techniques, see \emph{e.g.,}
\cite{chen2018rogueNLS, chen2018rogue, chen2019rogue}. From
\eqref{eqn:AKNSLaxPairPeriodict}, we see that the eigenfunctions are of the
form
\begin{align}
  &\Phi(x,t) = e^{\Omega t}y_1(x)
  \begin{pmatrix} -\hat B(x) \\ A(x)-\Omega \end{pmatrix},&
    &\text{or}&
    &\Phi(x,t)= e^{\Omega t}y_2(x)
  \begin{pmatrix} A(x)+\Omega \\ \hat C(x)\end{pmatrix}.&
  \label{eqn:theAKNSEigenfunctions}
\end{align}
Here the scalar functions $y_1(x)$ and $y_2(x)$ are determined by the
requirement that $\Phi(x,t)$ not only solves \eqref{eqn:AKNSLaxPairPeriodict},
but also \eqref{eqn:AKNSLaxPairPeriodicx}, since
(\ref{eqn:AKNSLaxPairPeriodic}a-b) have a common set of eigenfunctions. We will
use the first equation of \eqref{eqn:theAKNSEigenfunctions}, but both
representations can be helpful (see
Section~\ref{sec:ComputingLaxSpectrumGeneral} for more on this). Substitution
in
\eqref{eqn:AKNSLaxPairPeriodicx} gives
\begin{align}
     &-\hat B y_1' - \hat B_x y_1 = (-\alpha \hat B + Q(A-\Omega))y_1,& &(A-\Omega)y_1' + A_x
     y_1' = (-\kappa Q \hat B - \alpha(A-\Omega))y_1,
\end{align}
so that different (but equivalent) representations for $y_1(x)$ are obtained
from the first or second equation of~\eqref{eqn:AKNSLaxPairx}:
\begin{align}
  &y_1 = \hat y_1 \exp\lt(\int\frac{\alpha \hat B - Q (A-\Omega)-\hat
  B_x}{\hat B}~\d x\rt),&
  &y_1 = \tilde y_1 \exp\lt(-\int \frac{\kappa Q \hat B +\alpha (A-\Omega) +
  A_x}{A-\Omega}~\d x\rt),&\label{eqn:AKNSEigenfunctions}
\end{align}
where $\hat y_1$ and $\tilde y_1$ are constants of integration. At this point
we define $S_L$, the function space that defines $\sigma_L$. Physically, the
eigenfunctions $\phi$ should be bounded for $\overline{\R} = \R\cup\{\infty\}$.
Therefore we define $S_L$ to be the space where $\phi$ are bounded for
$x\in\overline{\R}$. Since $A$ and $B$ are bounded for $x\in \overline{\R}$
(Remark~\ref{remark:recursionOp}), $y_1(x)$ must be bounded for $x\in
\overline{\R}$ in order for $\phi(x,t)$ to be an eigenfunction. We use the
second expression of \eqref{eqn:AKNSEigenfunctions}.  Similar work is done for
the first expression of \eqref{eqn:AKNSEigenfunctions} in
Section~\ref{sec:ComputingLaxSpectrumGeneral}.  To bound the exponential
growth, we consider the real part of the exponential.  Therefore we need the
indefinite integral 
\begin{align}
  I &= \RE \int\lt( \frac{\kappa Q\hat B}{A-\Omega} +\alpha + \D{x}\log(A-\Omega)\rt)~\d
  x,\label{eqn:AKNSIntegralConditions}
\end{align}
to be bounded for $x\in\overline{\R}$. 
By Remark~\ref{remark:recursionOp} the integrand in $I$ is $P$-periodic, so it
suffices to examine the average over one period. Therefore we need
\begin{align}
  J &= \RE \Langle \alpha + \frac{\kappa Q
  \hat B}{A-\Omega} + \D{x}\log(A-\Omega)\Rangle = 0,
\label{eqn:AKNSIntegralConditionPeriodic}
\end{align}
where $\langle \cdot\rangle = \frac1P \int_0^P \cdot ~\d x$ is
the average over a period. Since $A-\Omega$ is $P$-periodic, the logarithmic
derivative has no contribution to $J$. \te{Therefore $\zeta \in \sigma_L$ if
and only if
\begin{align}
J = \RE\Langle \alpha + \frac{\kappa Q \hat B}{A-\Omega} \Rangle = 0.\label{eqn:AKNSLaxConditionNonreal}
\end{align}
\begin{remark}\label{remark:asymptotics}
The condition \eqref{eqn:AKNSLaxConditionNonreal} is a necessary and sufficient
condition for $\zeta \in \sigma_L$. In general, this condition is nontrivial to
work with. However, it does allow one to characterize large subsets of the Lax spectrum
with relative ease. In particular, the large $\zeta$ asymptotic analysis of
\eqref{eqn:AKNSLaxConditionNonreal} is tractable, in that we can find an
asymptotic approximation to the unbounded components of the Lax spectrum, in
anology to the asymptotic analysis used to find the essential spectrum for
decaying solutions on the real line \cite{kapitula_promislow}. Below we
examine \eqref{eqn:AKNSLaxConditionNonreal} for $\zeta \in \R$. 
\end{remark}
}
When $\zeta \in \R$, $\alpha \in i\R$ so it also has no contribution to $J$.
Therefore when $\zeta \in \R$, we need
\begin{align}
   \RE \Langle \frac{\kappa Q \hat B}{A-\Omega}\Rangle=0.
  \label{eqn:AKNSLaxCondition}
\end{align}
When using the recursion operator, $A(\R)\subset i\R$ since $A(\zeta)$ is
defined by a power series in $\zeta$ with imaginary coefficients (see
Remark~\ref{remark:recursionOp}). We have the following lemma for $\zeta\in
\R$. 

\begin{lemma}\label{lem:AKNSCB}
  Consider a member of the AKNS hierarchy with $r = \kappa q \cc$, where
  $\kappa = \pm 1$, and $q_t = r_t = 0$. Then $C = \kappa B\cc$ when
  $\zeta \in\R$ and $A(\zeta)\in i\R$, where $A,~B,$ and $C$ are defined in
  \eqref{eqn:AKNSLaxPairt}.
\end{lemma}

\begin{proof}
  We first establish that $\hat C  = \kappa \hat B \cc$
  \eqref{eqn:AKNSHats} from which it follows that $C = \kappa B\cc$. From
  \eqref{eqn:AKNScompatibilityHat},
  \begin{align}
    0 &= \RE(A_x) = Q \RE(\hat C - \kappa \hat B),\label{eqn:AKNSAx}
  \end{align}
  and
  \begin{align}
    0 &= \RE(\hat C_x - \kappa \hat B_x) = \RE[-2\alpha(\hat C + \kappa \hat B)
    + 4\kappa Q A] = \RE[-2\alpha(\hat C + \kappa \hat B)],\label{eqn:AKNSCxBx}
  \end{align}
  since $\kappa Q A \in i\R$ when $\zeta \in \R$, by assumption. If
  $\alpha(\zeta)\neq 0$, $\IM(\hat C) = -\kappa \IM(\hat B)$ and
  \eqref{eqn:AKNSAx} implies that $\hat C = \kappa \hat B\cc$. If $\alpha = 0$,
  we have
  \begin{align}
    \partial_x(\hat C + \kappa \hat B) = \hat C_x + \kappa \hat B_x = 0.
  \end{align}
  Since $\IM(\hat C + \kappa \hat B)=0$ at $x$ with $\alpha \neq 0$, it must be
  zero for all $x$. Therefore $\hat C = \kappa \hat B\cc$ for all $x$. 
\end{proof}

We can now identify part of the Lax spectrum with imaginary (stable) elements of
the stability spectrum.

\begin{theorem}\label{thm:AKNSThm}
    Consider a member of the AKNS hierarchy \eqref{eqn:AKNSLaxPair} satisfying
    Assumption~\ref{assumption:stationary} and Remark~\ref{remark:recursionOp}.
    Assume further that $A,~\hat B,$ and $\hat C$ from
    \eqref{eqn:AKNSLaxPairPeriodic} are $P$-periodic. Let $\cQ = \{\zeta \in \R:
    A(\zeta)\in i\R\},~\cS = \{\zeta \in \R: \Omega(\zeta)\in i\R\}$ ($\cS$ for
    stable) and $\cU = \{\zeta\in\R: \Omega(\zeta)\in \R\}$ ($\cU$ for
    unstable). 
    
    When $\kappa = -1$, $\cQ \subset \cS$. When $\kappa = 1$, $\cQ
    \subset\cS\cup \cU$. Also, $\cQ\cap \cS\subset\sigma_L$, and if $\Omega$
    is shown to be surjective, $\Omega(\cQ\cap\cS)\subset\sigma_{\cL}\cap i\R$
    for $\kappa = \pm 1$.  In other words, all real $\zeta$ for which
    $\Omega(\zeta),A(\zeta)\in i\R$ are part of the Lax spectrum and map to
    stable elements of the stability spectrum.

\end{theorem}

\begin{proof}
   Let $\zeta \in \cQ$. By Lemma~\ref{lem:AKNSCB}
  \begin{align}
    \Omega^2(\zeta) = A^2 + \kappa \abs{B}^2 \in \R,
  \end{align}
  so $\cQ\subset \cS \cup \cU$. When $\kappa = -1$, $\Omega^2(\zeta)<0$ so
  $\cQ\subset \cS$. At this point, we are not assuming that $\zeta \in
  \sigma_L$. We show that this is the case by showing that $\tilde I_1 = 0$
  \eqref{eqn:AKNSLaxCondition}.  If $\Omega(\zeta)\in i\R$, then
  \begin{align}
   \begin{split}
    \RE \frac{Q\hat B}{A-\Omega} &= \frac12 \left(\frac{Q \hat
    B}{A-\Omega} + \frac{Q \hat B\cc}{A\cc - \Omega\cc}\rt)\\
    &= \frac12\left(\frac{Q\hat B - \kappa \hat C}{A-\Omega}\right) =
    -\frac{A_x}{2(A-\Omega)} = -\frac12\D{x}\log(A-\Omega).
   \end{split}
  \end{align}
  Since $A$ is assumed to be periodic, $\RE\tilde I_1 = 0$, establishing that
  $\zeta \in \sigma_L$ and that $\cQ\cap \cS \subset\sigma_L$. By definition,
  $\Omega(\cQ\cap \cS)\subset i\R\cap \sigma_{\cL}$ when $\Omega$ is surjective
  onto $\sigma_{\cL}$. 
\end{proof}

  \begin{remark}
    When using the recursion operator, $A(\R)\subset i\R$ (see statement just
    before Lemma~\ref{lem:AKNSCB}). Therefore when $\kappa = -1$,
    Theorem~\ref{thm:AKNSThm} implies that $\R \subset \sigma_L$ and
    $\Omega(\R)\subset i\R$ for all members of the AKNS hierarchy satisfying
    Assumption~\ref{assumption:stationary} and where
    Remark~\ref{remark:recursionOp} applies.

    When $\kappa = 1$, the spectral problem \eqref{eqn:AKNSLaxPair} is self
    adjoint so $\sigma_L\subset \R$. Therefore the assumption that $\zeta \in \R$
    is not actually an assumption for such problems. Theorem~\ref{thm:AKNSThm}
    establishes that all $\{\zeta \in \R: \Omega(\zeta)\in i\R\} \subset \sigma_L$.
    If one can establish that $\{\zeta \in \R:\Omega(\zeta)\in
    \R\setminus\{0\}\}\not\subset\sigma_L$ for a specific problem, then one has
    shown that the underlying solution is stable.
  \end{remark}

\begin{example*}[The NLS equation, continued]
   $ $ \newline
   \te{We begin by demonstrating the statement of Remark \ref{remark:asymptotics}. The
   condition \eqref{eqn:AKNSLaxConditionNonreal} can be studied asymptotically for
   large $\zeta$. For both the defocusing and the focusing NLS equations,
   $\Omega \sim i\zeta^2$ for large $\zeta$.  Since $A(\R)\subset
   i\R$, $A-\Omega \in i\R$ if $\zeta \in \mathbb{R}$. Further, $\alpha(\R) \in
   i\R$ and $\kappa Q \hat B \sim \kappa Q^2 \zeta \in \R$ for real
   $\zeta$. Therefore for large real $\zeta$, 
   \begin{align} 
      \alpha + \frac{\kappa Q \hat B}{A-\Omega} \in i\R
   \end{align}
   implying that $\zeta \in \sigma_L$. \te{This gives two unbounded components of
   $\sigma_L$: one as $\zeta \to \infty$ and another as $\zeta \to -\infty$.  Next
   we demonstrate that these two are the only unbounded components of $\sigma_L$.}
   For large $\zeta$, the $x$-equation of the Lax pair \eqref{eqn:AKNSLaxPairx}
   becomes diagonal, giving
   \begin{align}
   \phi \sim \begin{pmatrix} e^{-i\zeta x}\\ e^{i\zeta x}\end{pmatrix}.
   \end{align}
   The only large $\zeta$ which leaves $\phi$ bounded for all $x\in \R$ is
   $\zeta\in \R$. \te{Therefore the two unbounded components asymptotic to the
   real line as $\zeta \to \pm\infty$ found above are the only two unbounded
   components of the Lax spectrum.}
   }

   The defocusing NLS equation has $q = r\cc$ so the Lax pair is self adjoint
   and $\sigma_L \subset \R$. In \cite{bottman2011elliptic} the authors
   establish that $\{\zeta\in\R:\Omega(\zeta)\in \R\setminus\{0\}\}
   \not\subset\sigma_L$ by working with \eqref{eqn:AKNSLaxCondition} directly.
   Since $\sigma_L  = \{\zeta \in \R: \Omega(\zeta)\in i\R\}$ and
   $\Omega:\sigma_L\to\sigma_{\cL_{\text{NLS}}}$ is surjective, the elliptic
   solutions are stable (see Figure \ref{fig:defocNLS}). 

   The situation for the elliptic solutions of the focusing NLS equation is more
   complicated since the Lax pair is not self adjoint and $\sigma_L$ is not
   a subset of the real line.  To find the rest of $\sigma_L$, we work with
   \eqref{eqn:AKNSIntegralConditionPeriodic} by computing the integral in terms
   of elliptic functions and working with the result, but it may also be found
   using the Floquet Discriminant (see Appendix~\ref{sec:FloquetDiscriminant}).
   It was originally shown in \cite{deconinck2017stability} that $\R\subset
   \sigma_L$ by working with the complicated expression
   for~\eqref{eqn:AKNSIntegralConditionPeriodic} directly.  In \cite{upsalNLS}
   we establish that the set $\cE = \R \cup \{\zeta\in
   \C:\Omega(\zeta)=0\}$ are the only elements in $\sigma_L$ that map to
   $i\R\subset\sigma_{\cL}$ under $\Omega$. Everything else in the spectrum maps
   to unstable modes, hence the solutions are, in general, unstable. Special
   classes of perturbations exist, the \emph{subharmonic perturbations}, for which
   only members of $\cE$ are excited. Some of the elliptic
   solutions are stable with respect to this class of perturbations
   \cite{upsalNLS}. This demonstrates the power of what is established here: since
   we know what maps to stable elements of the stability spectrum, we can find the
   class of perturbations for which the solution is stable. See Figure
   \ref{fig:focNLS} for an example.

  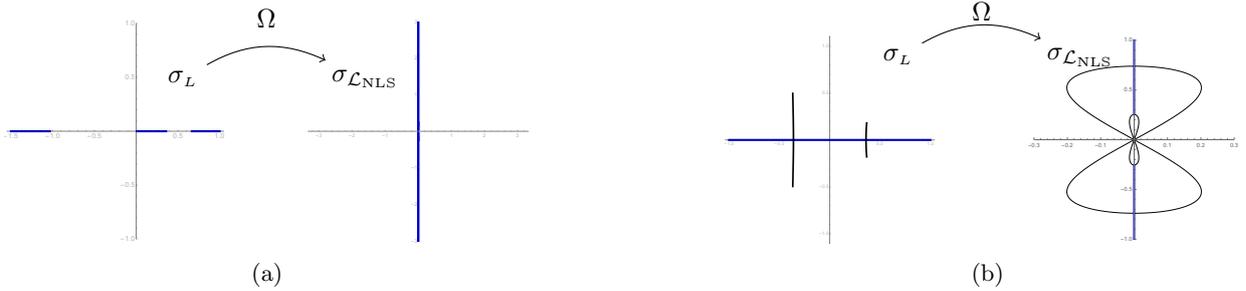
\begin{figure}
    \begin{subfigure}[b]{0.42\textwidth}
    \begin{tikzpicture}
       \node[inner sep=0pt] (defocLax) at (0,0)
       {\includegraphics[width=0.42\linewidth]{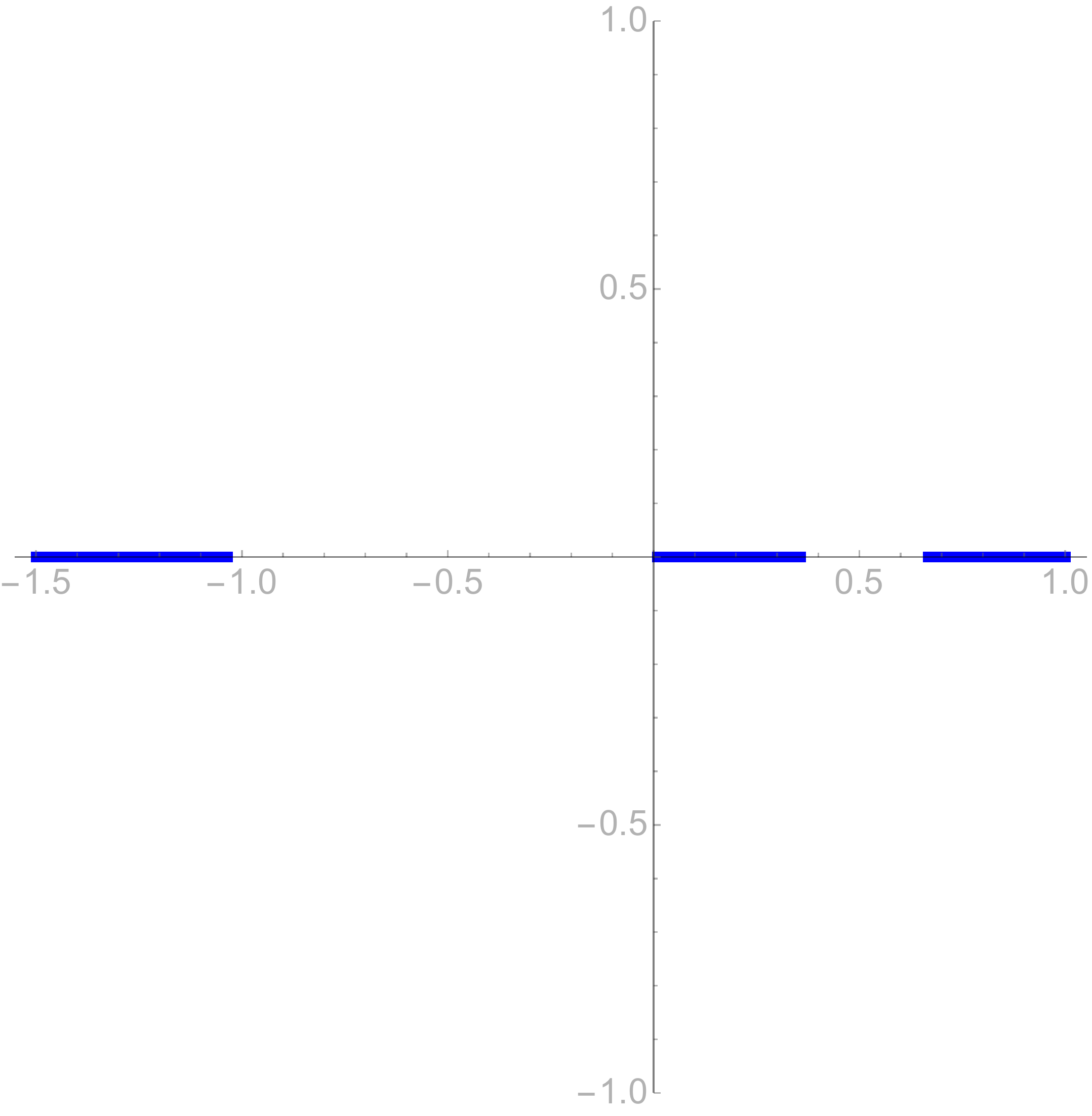}};
       \node[inner sep=0pt] (defocStab) at (4,0)
       {\includegraphics[width=0.42\linewidth]{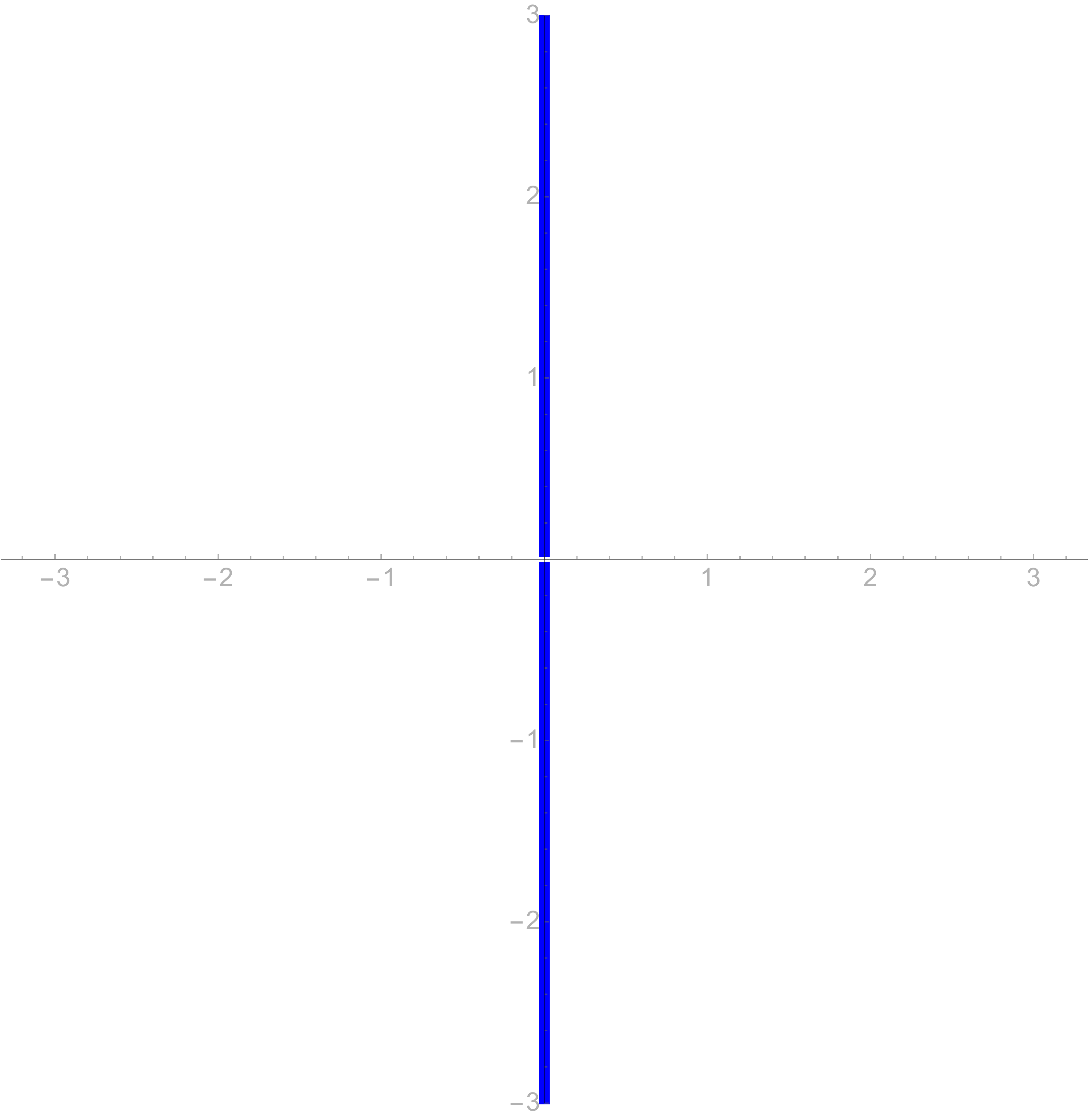}};
       \node (Lax) at (0.9,0.7) {$\sigma_{{\scriptscriptstyle L}}$};
       \node (Stab) at (3.3,0.7)
       {$\sigma_{\scriptstyle{\cL_{\text{\tiny{NLS}}}}}$};
       \draw [->] (Lax) to [out = 30, in=155] (Stab);
       \node at (2.0, 1.5) {$\Omega$};
    \end{tikzpicture}
      \caption{\label{fig:defocNLS}}
    \end{subfigure}
    \hspace{0.14\textwidth}
    \begin{subfigure}[b]{0.42\textwidth}
    \begin{tikzpicture}
       \node[inner sep=0pt] (focLax) at (0,0)
       {\includegraphics[width=0.4\linewidth]{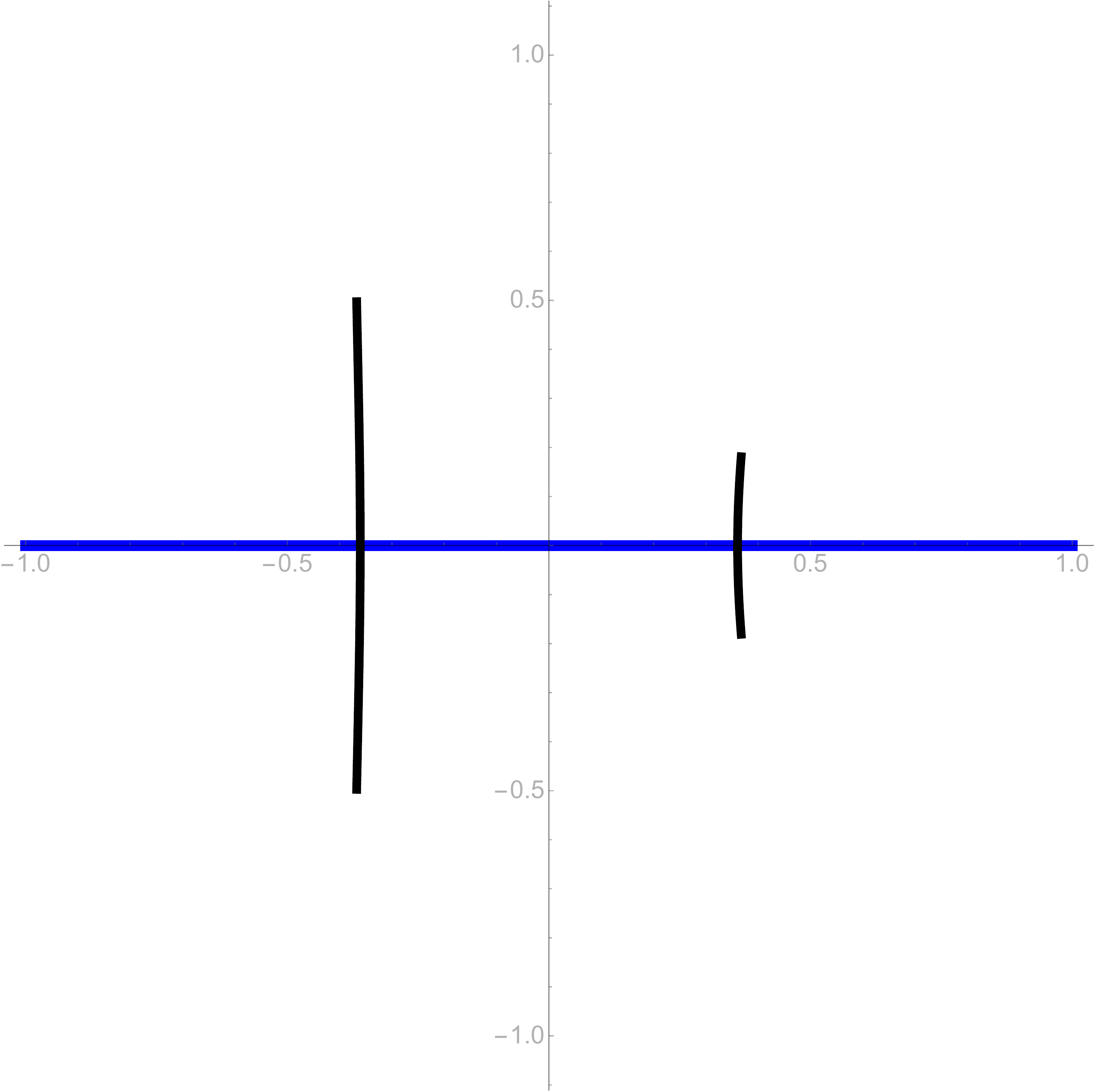}};
       \node[inner sep=0pt] (focStab) at (4,0)
       {\includegraphics[width=0.4\linewidth]{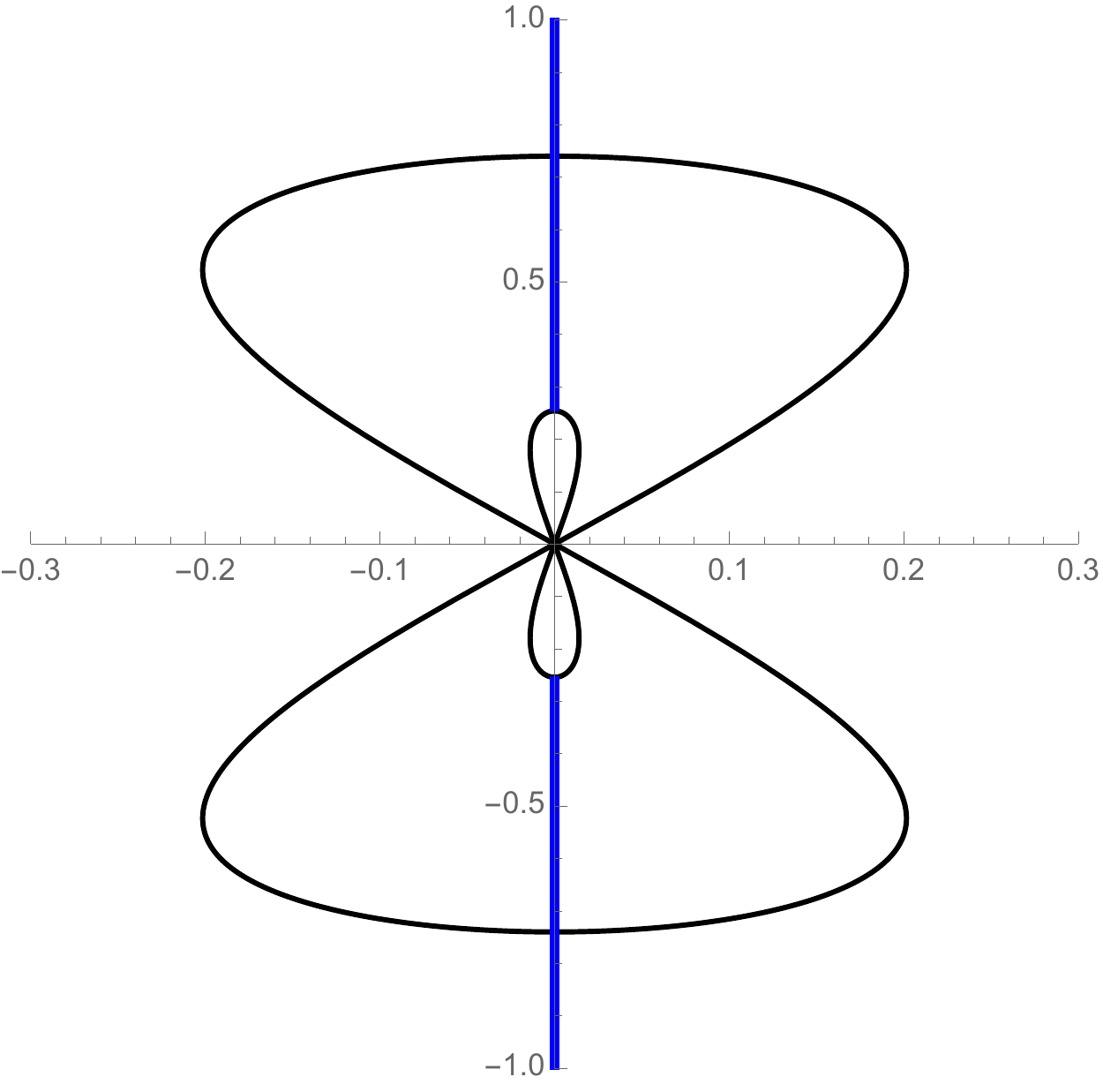}};
       \node (Lax) at (0.9,1.1) {$\sigma_{{\scriptscriptstyle L}}$};
       \node (Stab) at (3.3,1.1) {$\sigma_{\scriptstyle{\cL_{\text{\tiny{NLS}}}}}$};
       \draw [->] (Lax) to [out = 30, in=155] (Stab);
       \node at (2.0, 1.7) {$\Omega$};
    \end{tikzpicture}
      \caption{\label{fig:focNLS}}
    \end{subfigure}
    \caption{The real \vs imaginary part of the Lax and stability spectrum
      (left and right of each panel, respectively) for an elliptic solution of (a) the
      defocusing NLS equation and (b) the focusing NLS equation. The real
      component of $\sigma_{L}$ and its image under $\Omega$ is colored blue.
      The rest of $\sigma_L$ is in black. The Lax spectrum is computed
      analytically using \eqref{eqn:AKNSLaxCondition} \cite{bottman2011elliptic,
      deconinck2017stability} and the stability spectrum is the image under
      $\Omega$.  \label{fig:NLSspectrum}}
  \end{figure}
\end{example*}

\section{AKNS Examples}\label{sec:AKNSExamples}
In this section we apply the results of Section \ref{sec:AKNS} to three
integrable equations with connections to the AKNS hierarchy. The first equation
is in the AKNS hierarchy with the reduction \eqref{eqn:involutions}:
Theorem~\ref{thm:AKNSThm} applies directly.  Some results are new and others
confirm known results.  The second equation is in the AKNS hierarchy but does
not have the reduction \eqref{eqn:involutions}. Theorem~\ref{thm:AKNSThm} does
not apply immediately, but similar conclusions are made. The third equation is
in the AKNS hierarchy and \eqref{eqn:involutions} holds. However, a change of
variables is required to make the problem physical that takes it out of the
AKNS hierarchy.  Still, we find results identical to Theorem~\ref{thm:AKNSThm}.
This leads us to search for a generalization of Theorem~\ref{thm:AKNSThm} in
Section
\ref{sec:generalization}. 

\subsection{The modified Korteweg-de Vries equation}\label{example:mKdV}
The modified Korteweg-de Vries (mKdV) equation is given by
\begin{align}
  u_t - 6 \kappa u^2 u_x + u_{xxx} &= 0,\label{eqn:mKdV}
\end{align}
where $u$ is a real-valued function and $\kappa = -1$ and $\kappa = 1$
correspond to the focusing and defocusing cases respectively. Equation
\eqref{eqn:mKdV} is a member of the AKNS hierarchy (Section \ref{sec:AKNS})
with the Lax pair \cite{AKNS74}
  \begin{align}
   &A = -4i\zeta^3 - 2i\zeta qr, &
   &B = 4\zeta^2 q + 2q^2r + 2i\zeta q_x - q_{xx},&
   &C = 4\zeta^2 r + 2qr^2 - 2i\zeta r_x - r_{xx},&
  \end{align}
and $r = \kappa q = \kappa u$.  Letting $(y,\tau) = (x-ct, t)$, where $c\in\R$
is constant, gives mKdV in the traveling frame,
\begin{align}
  u_\tau - c u_y - 6 \kappa u^2 u_y + u_{yyy} &=0.\label{eqn:mKdVStationary}
\end{align}
The Lax pair for \eqref{eqn:mKdVStationary} changes accordingly:
\begin{align}
  &\phi_y = X\phi, & &\phi_\tau = (T+cX)\phi.&
\end{align}
The elliptic solutions of mKdV are found upon equating $u_\tau = 0$
\cite{mkdvOrbitalStability}. The Lax pair can be found by using the recursion
operator as mentioned in Remark~\ref{remark:recursionOp}, so the remarks there
holds and Assumption~\ref{assumption:stationary} applies.

Theorem~\ref{thm:AKNSThm} applies with $\cQ = \R$ since $A(\R)\subset i\R$.
When $\kappa = -1$, $\R\subset \sigma_L$ and $\Omega(\R)\subset
i\R\cap\sigma_{\cL}$. This result is new: it can be used as a first
step for studying the stability of the elliptic solutions of the focusing mKdV
equation. Solutions are stable with respect to perturbations that excite only
real elements of the Lax spectrum. We do not yet know what class of
perturbations this corresponds to. When $\kappa = 1$, $\cS = \{\zeta \in \R:
\Omega(\zeta)\in i\R\}\subset \sigma_L$, $\Omega(\cS)\subset i\R\cap
\sigma_{\cL}$, and $\Omega(\cS)\subset\sigma_{\cL}$. Since $\R\subset
\sigma_L$, one must establish that $\{\zeta\in\R:\Omega(\zeta)\notin
i\R\setminus\{0\}\}\not\subset \sigma_L$ to establish stability. This is proven in
\cite{mkdvOrbitalStability} to show that the elliptic solutions for the
defocusing mKdV equation are stable.

\subsection{The PT-symmetric reverse space nonlocal NLS
equation}\label{sec:NonlocalAKNS}

The PT-symmetric reverse space nonlocal NLS equation is given by
\cite{ablowitz-nonlocalNLS}
\begin{align}
  i \Psi_t(x,t) + \frac12 \Psi_{xx}(x,t) - \kappa \Psi(x,t)^2
  \Psi\cc(-x,t) &= 0,\label{eqn:NLNLS1}
\end{align}
where $\kappa = \pm 1$. This equation is a member of the AKNS hierarchy
(Section \ref{sec:AKNS}) with the Lax pair \cite{ablowitz-nonlocalNLS}
\begin{align}
  \label{eqn:NLNLSLaxPair}
  &A(x,t) = -i\zeta^2 - iqr/2,& &B = \zeta q + i q_x/2,& &C = \zeta r - i
  r_x/2,&
\end{align}
and $r(x,t) = \kappa q\cc(-x,t) = \kappa \Psi\cc(-x,t)\in \C$. Solutions of
\eqref{eqn:NLS} with
\begin{align}
  \Psi(x,t) = \Psi(-x,t),
\end{align}
are also solutions of \eqref{eqn:NLNLS1}. Thus all even solutions of
\eqref{eqn:NLS} examined in \cite{bottman2011elliptic} for $\kappa = 1$ and in
\cite{upsalNLS} for $\kappa = -1$ are solutions to \eqref{eqn:NLNLS1}.
These solutions, and other periodic and quasi-periodic solutions of
\eqref{eqn:NLNLS1}, were first reported in \cite{periodicSolnsOfNonlocalNLS1}.
Almost every solution found in \cite{periodicSolnsOfNonlocalNLS1} is even in
$x$, except for one that is odd. For the even solutions, the Lax spectrum
remains unchanged since $\Psi(x,t) = \Psi(-x,t)$: hence the stability results in
\cite{bottman2011elliptic, upsalNLS} hold for these solutions. 

The Lax pair is not self adjoint for $\kappa = 1$, so $\sigma_L$ is not
necessarily a subset of $\R$. With the equation written in a uniformly rotating
frame so that all components are time independent,
Assumption~\ref{assumption:stationary} and Remark~\ref{remark:recursionOp}
applies.  Upon assuming that $q$ and $r$ are $P$-periodic or quasiperiodic and
$\zeta\in\R$, we have that $\zeta \in \sigma_L$ if
\begin{align}
  \RE \Langle \frac{qC}{A-\Omega}\Rangle = 0.\label{eqn:nonlocalCondition}
\end{align}
This condition is found just as \eqref{eqn:AKNSLaxCondition} was found and is an
example of several equivalent conditions for the Lax spectrum derived in
Section~\ref{sec:ComputingLaxSpectrumGeneral}. For $\zeta \in \R$, 
\begin{align}
  A\cc(-x) = -A(x), \qquad \text{and}\qquad C\cc(-x) = \kappa
  B(x).\label{eqn:reverseSpaceParitySymmetry}
\end{align}
so that
\begin{align}
\begin{split}
  \Omega^2(\zeta) &= A(x)^2 + B(x)C(x) = A(-x)^2 + B(-x)C(-x)\\
  &= (A^2 + BC)\cc = \lt(\Omega^2(\zeta)\rt)\cc,
\end{split}
\end{align}
and $\Omega(\R) \subset \R \cup i\R$. If $\zeta \in \R$ and $\Omega(\zeta) \in
i\R$, 
\begin{align}
  \begin{split}
  \RE\int_0^P \lt(\frac{q(x)C(x)}{A(x)-\Omega}\rt)~\d x &=
  \frac12 \lt(\int_0^P \frac{q(x)C(x)}{A(x)-\Omega}~\d x +
                  \int_0^P \frac{q\cc(x) C\cc(x)}{A\cc(x)-\Omega\cc}~\d x
                  \rt)\\
  &= \frac12 \lt(\int_0^P \frac{q(x)C(x)}{A(x)-\Omega}~\d x -
                  \int_0^{-P} \frac{q\cc(-x) C\cc(-x)}{A\cc(-x)-\Omega\cc}~\d x
                  \rt)\\
  &= \frac12 \lt(\int_0^P \frac{q(x)C(x)}{A(x)-\Omega}~\d x +
                  \int_{-P}^{0} \frac{q\cc(-x) C\cc(-x)}{-A(x)+\Omega}~\d x
                  \rt)\\
  &= \frac12 \lt(\int_0^P \frac{q(x)C(x)}{A(x)-\Omega}~\d x +
                  \int_{0}^{P} \frac{q\cc(-x) C\cc(-x)}{-A(x)+\Omega}~\d x
                  \rt)\\
  &= \frac12 \lt(\int_0^P \frac{q(x)C(x) -
                    q\cc(-x)C\cc(-x)}{A(x)-\Omega}~\d x\rt)\\
  &= \frac12 \int_0^P \frac{A_x}{A(x)-\Omega}~\d x =0,
  \end{split}
\end{align}
since $A$ is periodic when $q$ and $r$ are. It follows from
\eqref{eqn:nonlocalCondition} that $\{\zeta \in \R : \Omega(\zeta)\in i\R\}
\subset \sigma_L$ for $\kappa = \pm 1$.

\subsection{The sine- and sinh-Gordon equations \label{sec:SG}}
The sine-Gordon (s-G) equation in light-cone coordinates is given by
\begin{align}
  u_{\xi \eta} &= \sin u,\label{eqn:SGLC}
\end{align}
where $u(\xi,\eta)$ is real valued. Equation \eqref{eqn:SGLC} is a member of the
AKNS hierarchy (Section \ref{sec:AKNS}) with Lax pair \cite{AKNS74}
\begin{align}
  &A = \frac{i}{4\zeta}\cos(u),& \qquad &B = \frac{i}{4\zeta}\sin(u),& \qquad &C =
  \frac{i}{4\zeta} \sin(u).&
\end{align}
We write $(\xi, \eta)$ instead of $(x,t)$ to distinguish between the light-cone
coordinates $(\xi,\eta)$ and the space-time coordinates $(x,t)$. Equation
\eqref{eqn:SGLC} is equivalent to the compatibility of mixed derivatives,
$\partial_\eta v_\xi = \partial_\xi v_\eta$, by requiring that $r = -q =
u_\xi/2$. Since $r = -q$, \eqref{eqn:SGLC} is not self adjoint. A self-adjoint
variant of the s-G equation is the sinh-Gordon (sh-G) equation,
\begin{align}
  u_{\xi \eta} &= \sinh u.\label{eqn:sinhGLC}
\end{align}
Equation \eqref{eqn:sinhGLC} is a member of the AKNS hierarchy (Section
\ref{sec:AKNS}) with Lax pair \cite{AKNS74}
\begin{align}
  A &= \frac{i}{4\zeta} \cosh(u), \qquad B = -\frac{i}{4\zeta}\sinh(u), \qquad C
  = \frac{i}{4\zeta}\sinh(u),
\end{align}
and is equivalent to the compatibility of mixed derivatives under the reduction
$r = q = u_\xi/2$.

The Lax par for both the s-G and the sh-G equations can be found by using the
recursion operator as mentioned in Remark~\ref{remark:recursionOp}, so the
remarks there and Assumption~\ref{assumption:stationary} apply.
Theorem~\ref{thm:AKNSThm} applies to the s-G equation when $\kappa = -1$ and to
the sh-G equation when $\kappa = 1$. For $\kappa = -1$, $\R\subset\sigma_L$ and
$\Omega(\R)\in \sigma_{\cL}\cap i\R$. For $\kappa = 1$,
$\R\cap\{\zeta\in\R:\Omega(\zeta)\in i\R\} \subset\sigma_L$. However since
$\eta$ is not a timelike variable, stability results mean little for this
equation in these variables. Instead, we transform from light-cone to
laboratory coordinates.
To transform \eqref{eqn:SGLC} from light-cone to laboratory coordinates, we let
$(x,t) = (\eta + \xi, \eta - \xi)$ to obtain 
\begin{align} 
	u_{tt} - u_{xx} + \sin(u) &= 0.\label{eqn:SG}
\end{align}
The same coordinate transformation on \eqref{eqn:sinhGLC} gives
\begin{align}
  u_{tt} - u_{xx} + \sinh(u) &= 0.\label{eqn:ShG}
\end{align}
The Lax pair for both systems is
\begin{align}
  w_x &= \frac12(T+X) w = \hat Xw, \qquad w_t = \frac12(T-X)w = \hat T w.\label{eqn:SGLaxPair}
\end{align}
Note that \eqref{eqn:SG} and \eqref{eqn:ShG} are not members of the AKNS
hierarchy. Nonetheless we show that statements very similar to those made in
Theorem~\ref{thm:AKNSThm} hold for this equation. This gives us a bridge from
the AKNS framework to generalizations.

We move the s-G equation \eqref{eqn:SG} and the sh-G equation \eqref{eqn:ShG}
to a traveling frame by letting \mbox{$(z,\tau)=(x-Vt,t)$} for constant $V\in
\R$ and find
\begin{align}
  (V^2-1) u_{zz} - 2V u_{z\tau} + u_{\tau\tau}+\sin(u) = 0,
\end{align}
and
\begin{align}
  (V^2-1) u_{zz} - 2V u_{z\tau} + u_{\tau\tau} + \sinh(u) = 0,
\end{align}
respectively. The new Lax pair is given by
\begin{align}
  w_z &=\hat X w, \qquad w_\tau = (\hat T + V \hat X) w = \tilde T
  w.\label{eqn:SGLaxPairTravel}
\end{align}
Periodic stationary solutions are found by letting $u_{\tau} = 0$ and assuming
$u$ is periodic in $x$. Assumption~\ref{assumption:stationary} holds and the Lax
pair is of the form mentioned in Remark~\ref{remark:recursionOp}. Since
the transformation to lab coordinates is isospectral, the spectrum $\sigma_L$
does not change:  $\R\subset \sigma_L$ and $\Omega(\R)\subset i\R\cap
\sigma_{\cL}$ (the squared-eigenfunction connection gives $\lam =
2\Omega(\zeta)$ here as well \cite{SGstability}). Solutions are stable with
respect to perturbations that excite only real elements of the Lax spectrum.
These results are known and have been shown in \cite{SGstability} (whose
results agree with \cite{JonesMarangell-sG} where the Lax spectrum is not used).
For the sh-G equation, $\cS = \{\zeta \in \R:\Omega(\zeta)\in
i\R\}\subset\sigma_L$, $\Omega(\cS)\subset i\R\cap \sigma_{\cL}$, and
$\Omega(\cS)\subset \sigma_{\cL}$.  Since $\R\subset\sigma_L$, one must
establish that $\cU = \{\zeta\in\R:\Omega(\zeta)\in
i\R\setminus\{0\}\}\not\subset\sigma_L$ to establish stability. This result is
new and $\cU\not\subset\sigma_L$ has not yet been shown. 

The above results show that Theorem~\ref{thm:AKNSThm} may be applicable to some
non-AKNS integrable equations. In the next section, we establish a theorem to
this effect.  However, before continuing on we continue to study the spectral
problem for \eqref{eqn:SG} and \eqref{eqn:ShG} for exposition purposes.

The Lax pair \eqref{eqn:SGLaxPair} defines a quadratic eigenvalue problem (QEP),
\begin{align}
  Q(\zeta) w &= (M\zeta^2 + N\zeta + K) w = 0. \label{eqn:QEP}
\end{align}
There are two choices for $M,~N$, and $K$:
\begin{subequations}
\begin{align}
  &M_1 = \begin{pmatrix} 1 & 0 \\ 0 &1 \end{pmatrix},&
  &N_1 = \begin{pmatrix} -2 i\partial_x & iq \\ -ir &
    2i\partial_x\end{pmatrix},&
  &K_1 = \begin{pmatrix} ia & ib \\ ic & -ia \end{pmatrix},&\\
  &M_2 = \begin{pmatrix} 1 & 0 \\ 0 &-1 \end{pmatrix},&
  &N_2 = \begin{pmatrix} -2 i\partial_x & iq \\ ir &
    -2i\partial_x\end{pmatrix},&
  &K_2 = \begin{pmatrix} ia & ib \\ -ic & ia \end{pmatrix},
\end{align}
\end{subequations}
where
\begin{align}
  a = \zeta A , \quad b = \zeta B, \quad c = \zeta C.
\end{align}
Then $\zeta\in \C$ is an eigenvalue of $Q$ if and only if $Q(\zeta)w = 0$ for
all bounded $w$. A QEP is classified as self adjoint if $M,~N,$ and $K$ are
self adjoint \cite{tisseur2001quadratic}. The eigenvalues for self-adjoint QEPs
are either real or come in complex-conjugate pairs. If $M_1,~N_1,$ and $K_1$
are chosen, then $Q(\lambda)$ is self adjoint if $r = q\cc$, $a\cc = -a$, and
$c\cc = b$, which is the case for the sh-G equation. If $M_2,~N_2,$ and $K_2$
are chosen, then $Q(\lambda)$ is self adjoint if $r = -q\cc$, $a\cc = -a$, and
$c\cc = -b$, which is the case for the s-G equation.  It follows that for
either equation, the Lax spectrum consists of real or complex-conjugate
spectral elements. This confirms what we know from the isospectral transform to
light-cone coordinates.  The whole real line is part of the Lax spectrum for
the s-G equation and the Lax spectrum for the sh-G equation is a subset of the
real line. To determine the subset of $\sigma_L$ off the real line, one may use
the integral condition \eqref{eqn:AKNSIntegralConditionPeriodic} (used in
\cite{SGstability} to find $\sigma_L$) or the Floquet discriminant (Section
\ref{sec:FloquetDiscriminant}).

\section{Generalization of the AKNS results}\label{sec:generalization}
As seen from the AKNS examples, the real line of the Lax spectrum plays an
important role in stability. For the non self-adjoint members of
the AKNS hierarchy (those with $\kappa = -1$), the real line is part of the Lax
spectrum and maps to stable elements of the stability spectrum. For the
self-adjoint members of the hierarchy (those with $\kappa = 1$), the real line
is a subset of the Lax spectrum. If one can establish that
$\{\zeta\in\R:\Omega(\zeta)\in \R\setminus\{0\}\}\not\subset \sigma_L\}$, then
the solution of interest is stable. The sine-Gordon and sinh-Gordon examples
indicate that this trend holds even for integrable systems not in the AKNS
hierarchy. In what follows, we extend the AKNS results to integrable systems
that are not in the AKNS hierarchy.

\subsection{Setup}

In this section we consider integrable equations \eqref{eqn:generalNLEE}
possessing a $2\times 2$ Lax pair of the form,
\begin{subequations}
\label{eqn:generalLaxPair}
\begin{align}
\label{eqn:generalLaxPairx}
\phi_x(x,t;\zeta) &= \begin{pmatrix*}[r]
    \alpha(x,t;\zeta) & \beta(x,t;\zeta) \\
    \gamma(x,t;\zeta) & -\alpha(x,t;\zeta)
  \end{pmatrix*} \phi(x,t;\zeta) = X\phi,\\
\label{eqn:generalLaxPairt}
\phi_t(x,t;\zeta) &= \begin{pmatrix*}[r]
    A(x,t;\zeta) & B(x,t;\zeta) \\
    C(x,t;\zeta) & -A(x,t;\zeta)
  \end{pmatrix*}\phi(x,t;\zeta) = T\phi,
\end{align}
\end{subequations}
where $\alpha,~\beta,~\gamma,~A,~B,$ and $C$ are complex-valued functions.  
As in our analysis of the AKNS hierarchy, we restrict our analysis to Lax pairs
where the elements of $\cP = \lt\{\alpha,~\beta,~\gamma,~A,~B,~C\rt\}$ are
bounded for all $x\in \overline{\R}$ and are autonomous in $t$. Since we are interested in
studying stationary solutions, we assume again that $\alpha_t = \beta_t =
\gamma_t = 0$.  In the stationary frame, the compatibility of
\eqref{eqn:generalLaxPair} defines the conditions
\begin{subequations}
  \label{eqn:compatibilityConditionsStationary}
\begin{align}
  A_x &= \beta C - \gamma B ,\label{eqn:Aeqn}\\
  B_x &= 2\lt(\alpha B - \beta A \rt),\label{eqn:Beqn}\\
  C_x &= -2\lt(\alpha C - \gamma A \rt)\label{eqn:Ceqn}.
\end{align}
\end{subequations}
The definition for the Lax spectrum \eqref{eqn:LaxSpectrumDef} is unchanged.

\subsection{Computing the Lax spectrum}\label{sec:ComputingLaxSpectrumGeneral}
With $A$, $B$,  and $C$ $t$-independent, \eqref{eqn:generalLaxPairt} may be
solved by separation of variables resulting in once again \eqref{eqn:Omega}.
Here $\Omega$ has the same properties as for the AKNS hierarchy and
Lemma~\ref{lem:OmegaIndep} holds with a nearly identical proof. We again
consider the special case of the reduction
\begin{align}
  \gamma = \kappa \beta \cc,\qquad \kappa = \pm 1.\label{eqn:generalReduction}
\end{align}
With such a reduction, we let
\begin{align}
  \beta(x;\zeta) = \eta(x;\zeta) e^{i\theta(x;\zeta)}, \qquad
  \gamma(x;\zeta) = \kappa\eta(x;\zeta)  e^{-i\theta(x;\zeta)},
\end{align}
where $\eta$ and $\theta$ are real-valued functions with $\eta(x;\zeta)\geq 0$.
We also assume that $\eta$ is a $P$-periodic function.  Using the isospectral
transformation \eqref{eqn:Isospectral}, the Lax pair \eqref{eqn:generalLaxPair}
becomes
\begin{align}
  \Phi_x = \begin{pmatrix} \hat \alpha & \hat \beta \\ \hat \gamma & -\hat
    \alpha\end{pmatrix}\Phi, \qquad \Phi_t = \begin{pmatrix} \hat A & \hat B \\ \hat C
    & -\hat A\end{pmatrix}\Phi,\label{eqn:isospectralLaxPair}
\end{align}
where
  
\begin{equation}
\begin{aligned}
  &\hat \alpha = \alpha - i\theta_x/2, & &\hat \beta = \beta
  e^{-i\theta} = \eta,& &\hat \gamma = \gamma e^{i\theta} = \kappa\eta,&\\
  &\hat A = A,& &\hat B = e^{-i\theta} B,& &\hat C = e^{i\theta} C.&
\end{aligned}
\label{eqn:periodicLaxPair}
\end{equation}
The eigenfunctions here are identical to \eqref{eqn:theAKNSEigenfunctions}.
Similar to what we did there, we get two ODEs for $y_1$, but also two for $y_2$.
This gives two expressions for $y_1$ and two for $y_2$. In each case the
exponential term needs to be bounded for $x\in \overline{\R}$. Inspired by the
results for the AKNS hierarchy, we assume that $\hat A,~\hat B,$ and $\hat C$
are $P$-periodic. With these assumptions, we obtain eight boundedness
conditions that define the Lax spectrum (four from $y_1$ and $y_2$, and four
from rewriting those four using
\eqref{eqn:compatibilityConditionsStationary}):
\begin{equation}
\begin{aligned}
&\RE  \Langle\hat \alpha - \frac{\hat \beta(\hat A-\Omega)}{\hat B}\Rangle = 0,&
&\RE  \Langle\hat \alpha + \frac{\hat \gamma \hat B}{\hat A-\Omega}\Rangle = 0,&
&\RE  \Langle\hat \alpha + \frac{\hat \beta \hat C}{\hat A+\Omega}\Rangle = 0,&
&\RE  \Langle\frac{\hat \beta \Omega}{\hat B}\Rangle = 0,&\\
&\RE  \Langle\hat \alpha - \frac{\hat \gamma(\hat A+\Omega)}{\hat C}\Rangle = 0,&
&\RE  \Langle\hat \alpha + \frac{\hat \beta \hat C}{\hat A-\Omega}\Rangle = 0,&
&\RE  \Langle\hat \alpha + \frac{\hat \gamma \hat B}{\hat A+\Omega}\Rangle = 0,&
&\RE  \Langle\frac{\hat \gamma \Omega}{\hat C}\Rangle = 0,&
\end{aligned}
  \label{eqn:newLaxSpectrumConditions}
\end{equation}
If any of these conditions are satisfied for a particular $\hat \zeta \in \C$,
then $\hat\zeta\in\sigma_L$. Some of these conditions are new and some have
been used in \cite{bottman2009kdv, bottman2011elliptic, deconinck2010orbital,
SGstability, mkdvOrbitalStability, deconinck2017stability, upsalNLS}. This is
the first time all are written down in full generality. 

\te{
\begin{remark}\label{remark:asymptotics2}
The conditions \eqref{eqn:newLaxSpectrumConditions} are generalizations of the
Lax condition \eqref{eqn:AKNSLaxConditionNonreal} for members of the AKNS
hierarchy. Each condition is necessary and sufficient for $\zeta \in \sigma_L$.
This condition is typically nontrivial to work with directly, but it does allow
one to find large subsets of the Lax spectrum by considering large $\zeta$
asymptotics. We note that the last condition in
\eqref{eqn:newLaxSpectrumConditions} implies that $\{\zeta\in \C: \Omega(\zeta)
= 0\}\subset \sigma_L$.
\end{remark}
}

Lemma~\ref{lem:AKNSCB} and Theorem~\ref{thm:AKNSThm} have immediate analogues
here, whose proof is nearly identical.

\begin{theorem}\label{thm:nonAKNSThm}
  Consider an integrable equation \eqref{eqn:generalNLEE} possessing the Lax
  pair \eqref{eqn:generalLaxPair} with the reduction
  \eqref{eqn:generalReduction}. Assume that $\hat A,~\hat B,$ and $\hat C$ from
  \eqref{eqn:periodicLaxPair} are $P$-periodic. Let 
  \begin{equation}
  \begin{aligned}
    &\cQ_{-} = \{\zeta\in\C: \alpha(\zeta),A(\zeta)\in i\R \text{ and } \beta =
  -\gamma\cc\},&
  &\cQ_{+} = \{\zeta\in\C: \alpha(\zeta),A(\zeta)\in i\R \text{ and } \beta =
  \gamma\cc\},&\\
  &\hspace{2cm}\cS = \{\zeta \in \C: \Omega(\zeta)\in i\R\},\hspace{1cm} \text{and}&
  &\hspace{1cm}\cU = \{\zeta\in\C: \Omega(\zeta)\in \R\}.&
  \end{aligned}
  \end{equation}
  Then $\cQ_-\subset \cS,~\cQ_+\subset\cS\cup\cU,~\cQ_{\pm}\cap
  \cS\subset\sigma_L$, and $\Omega(\cQ_{\pm}\cap \cS)\subset\sigma_{\cL}\cap
  i\R$ if $\Omega$ is shown to be surjective.

  In other words, assume that $\alpha(\zeta),A(\zeta)\in i\R$. Then
  $\Omega(\zeta)\in i\R$ when $\beta(\zeta) = -\gamma(\zeta)\cc,$ (akin to the
  $\kappa = -1$ case in Theorem~\ref{thm:AKNSThm}), and those $\zeta$ are in the
  Lax spectrum, $\sigma_L$. Further, they map to stable elements of the
  stability spectrum. When $\beta(\zeta) = \gamma(\zeta)\cc,$ (akin to the
  $\kappa = 1$ case in Theorem~\ref{thm:AKNSThm}), $\Omega(\zeta)\in i\R\cup\R$
  and the $\zeta$ such that $\Omega(\zeta)\in i\R$ are in the Lax spectrum and
  map to stable elements of the stability spectrum.
\end{theorem}
The proof of this theorem is nearly identical to the proof of
Theorem~\ref{thm:AKNSThm}, so we omit it here. Before moving on to examples
which demonstrate the applicability of Theorem~\ref{thm:nonAKNSThm}, we remark
that our results do address the intersection of $\sigma_L$ with $\cQ_\pm\cap
\cU$. Indeed, establishing whether or not this is true is important for
understanding the stability of the solution in question.  We do not address
this difficulty here. 

\section{Examples}\label{sec:examples}
In this section provide examples for which Theorem~\ref{thm:nonAKNSThm}
applies. Next, we provide examples for which the theorem does not apply
directly, but for which similar conclusions can be drawn.  We do this to
demonstrate that these results can be generalized and are not necessarily
limited to the special cases considered here.

\subsection{Derivative nonlinear Schr\"{o}dinger equation}
The derivative NLS (dNLS) equation,
\begin{align}
  i q_t &= -q_{xx} +  i\kappa (\abs{q}^2 q)_x,\label{eqn:DNLS} \qquad \kappa =
  \pm 1,
\end{align}
was first solved on the whole line using the Inverse Scattering Transform in
\cite{kaupDNLS}.  The Lax pair for \eqref{eqn:DNLS} is
given by \eqref{eqn:generalLaxPair} with \cite{kaupDNLS}
\begin{equation}
  \label{eqn:DNLSLaxPair}
\begin{aligned}
  & \alpha = -i\zeta^2, & & \beta = q\zeta, & &\gamma = r\zeta,&\\
  &A = -2i\zeta^4 -i\zeta^2 r q, & &B = 2\zeta^3 q + i\zeta q_x + \zeta r q^2, &
  &C = 2\zeta^3 r - i\zeta r_x + \zeta r^2 q.&
\end{aligned}
\end{equation}
where $r = \kappa q\cc \in \C$.  Using $q(x,t) \mapsto e^{-i\omega t}q(x,t)$
where $\omega$ is a real constant, \eqref{eqn:DNLS} becomes
\begin{align}
  i q_t &= -q_{xx} + i\kappa(\abs{q}^2 q)_x - \omega q,
\end{align}
and $A \mapsto A + i\omega /2$; otherwise \eqref{eqn:DNLSLaxPair} remains the
same. Stationary solutions satisfy
\begin{align}
  -q_{xx} + i\kappa(\abs{q}^2 q)_x - \omega q = 0.
\end{align}
Quasi-periodic elliptic solutions to the stationary problem were found in
\cite{DNLSWaves}.

The Lax pair defines a QEP \eqref{eqn:QEP}. There are two choices for $M,~N,$
and $K$,
\begin{subequations}
\begin{align}
  &M_1 = \begin{pmatrix} 1 & 0 \\ 0 & 1 \end{pmatrix},&
  &N_1 = \begin{pmatrix} 0 & i q \\ -i r & 0 \end{pmatrix},&
  &K_1 = \begin{pmatrix} -i\partial_x & 0 \\ 0 & i\partial_x\end{pmatrix},&\\
  &M_2 = \begin{pmatrix} 1 & 0 \\ 0 & -1\end{pmatrix},&
  &N_2 = \begin{pmatrix} 0 & i q \\ ir & 0 \end{pmatrix},&
  &K_2 = \begin{pmatrix} -i\partial_x & 0 \\ 0 & -i\partial_x\end{pmatrix}.&
\end{align}
\end{subequations}
If $M_1,~ N_1,$ and $K_1$ are chosen, then $Q(\lambda)$ is self adjoint if $r =
q\cc$. If $M_2,~N_2,$ and $K_2$ are chosen, then $Q(\lambda)$ is self adjoint if
$r = -q\cc$. It follows that eigenvalues are real or come in complex-conjugate
pairs for either choice of $\kappa$.

Since Assumption~\ref{assumption:stationary} holds and
Remark~\ref{remark:recursionOp} applies here, Theorem~\ref{thm:nonAKNSThm}
applies with different results depending on $\kappa$ and $\zeta$. We have
$\zeta\in\cQ_-$ when $\zeta \in \R$ and $\kappa = -1$ or when $\zeta \in i\R$
and $\kappa = 1$.  Alternatively, $\zeta \in \cQ_+$ when $\zeta \in \R$ and
$\kappa = 1$ or when $\zeta \in i\R$ and $\kappa = -1$. Defining $\Omega_i =
\lt\{\zeta \in \C : \Omega(\zeta)\in i\R\rt\}$, $\R \cup (\Omega_i\cap i\R)
\subset \sigma_L$ and $\Omega(\R \cup (\Omega_i\cap i\R))\subset i\R$ for
$\kappa = -1$ (see Figure \ref{fig:focdNLS}). If $\kappa = 1$, $i\R \cup (\R
\cap \Omega_i)\subset \sigma_L$ and $\Omega(i\R \cup(\R \cap\Omega_i))\subset
i\R$ (see Figure \ref{fig:defocdNLS}).  To compute the spectrum off of the real
or imaginary axes, one must examine the integral conditions
\eqref{eqn:newLaxSpectrumConditions} or construct the Floquet discriminant
(Appendix \ref{sec:FloquetDiscriminant}).

  \begin{figure}
    \begin{subfigure}[b]{0.48\textwidth}
    \begin{tikzpicture}
       \node[inner sep=0pt, align=left] (defocLax) at (0,0)
       {\includegraphics[width=0.55\linewidth]{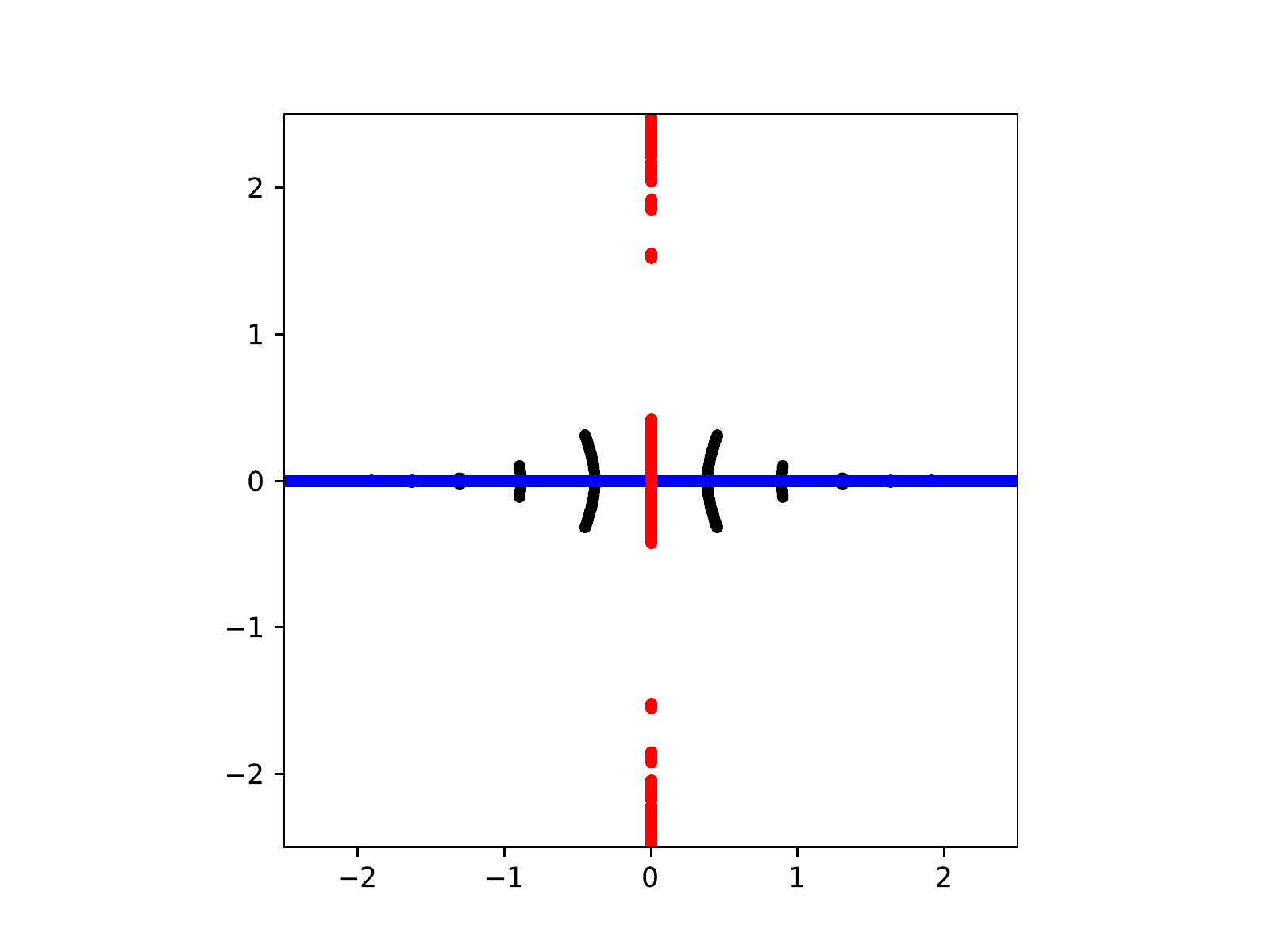}};
       \node[inner sep=0pt, align=left] (defocStab) at (3.5,0)
       {\includegraphics[width=0.55\linewidth]{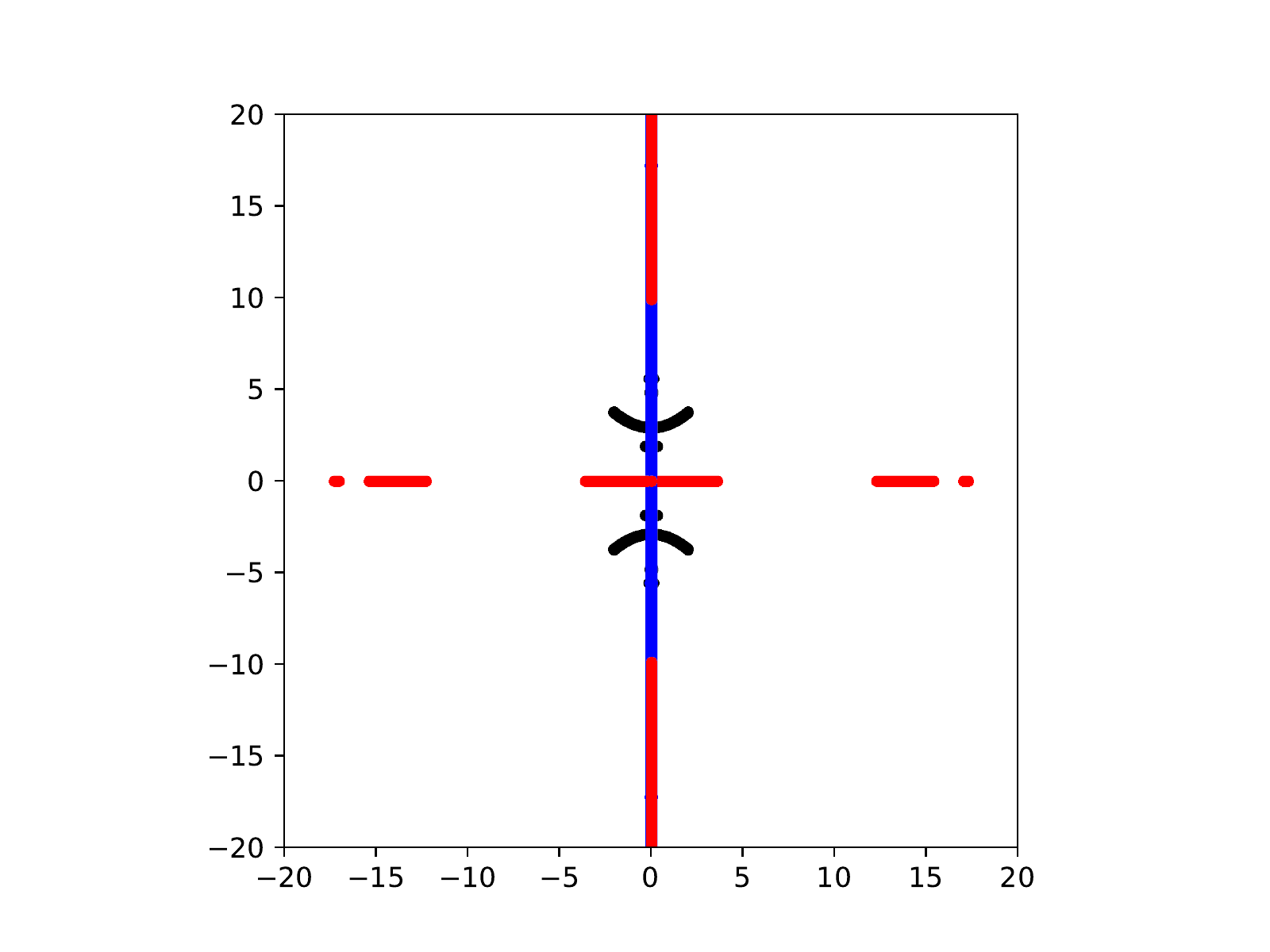}};
       \node (Lax) at (0.9,1.7) {$\sigma_{{\scriptstyle L}}$};
       \node (Stab) at (3.3,1.7)
       {$\sigma_{\scriptstyle{\cL_{\scriptscriptstyle{\text{dNLS}}}}}$};
       \draw [->] (Lax) to [out = 30, in=155] (Stab);
       \node at (2.0, 2.5) {$\Omega$};
    \end{tikzpicture}
      \caption{\label{fig:focdNLS}}
    \end{subfigure}
\hspace{0.02\textwidth}
    \begin{subfigure}[b]{0.48\textwidth}
    \begin{tikzpicture}
       \node[inner sep=0pt] (focLax) at (0,0)
       {\includegraphics[width=0.55\linewidth]{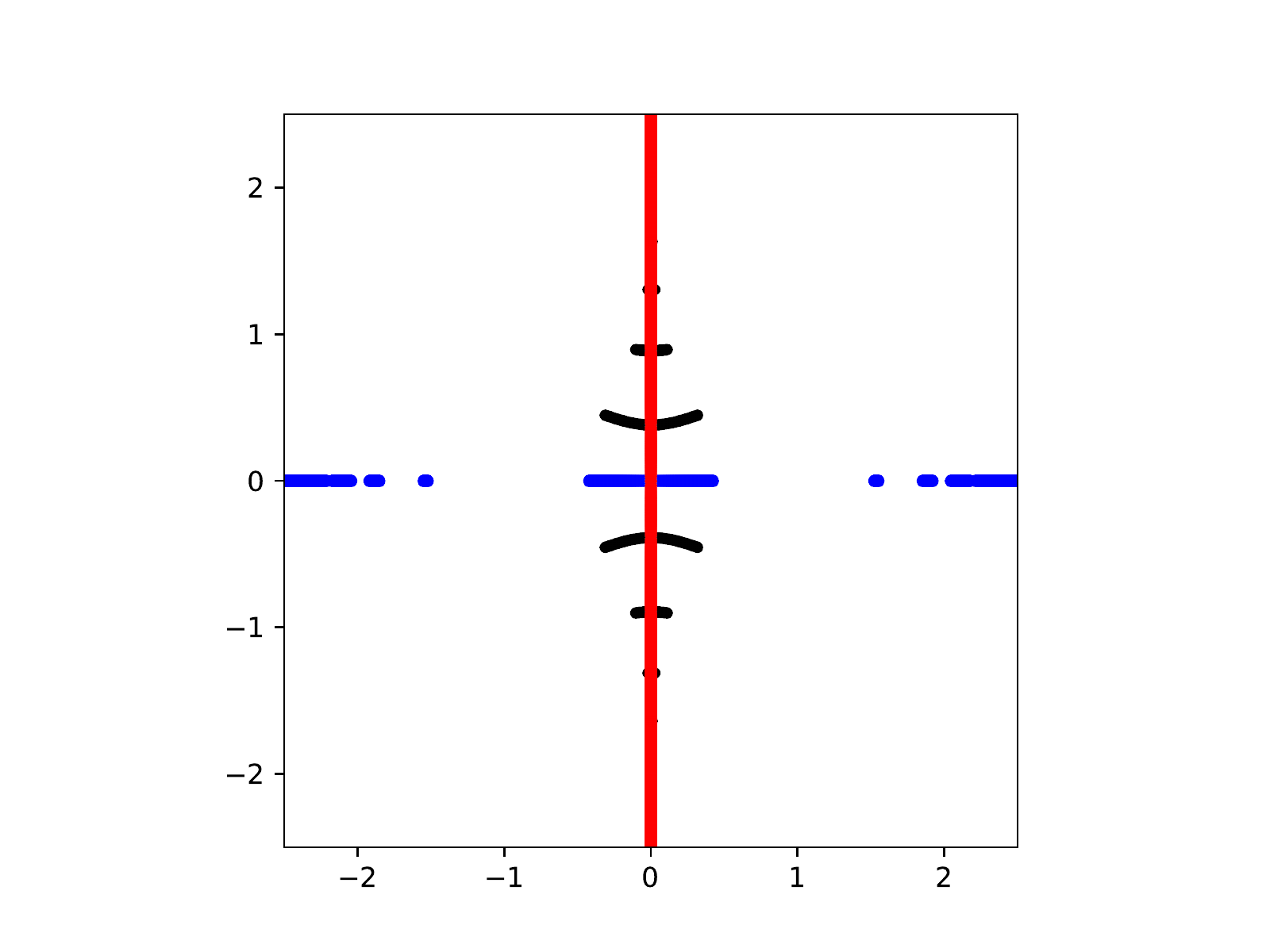}};
       \node[inner sep=0pt] (focStab) at (3.5,0)
       {\includegraphics[width=0.55\linewidth]{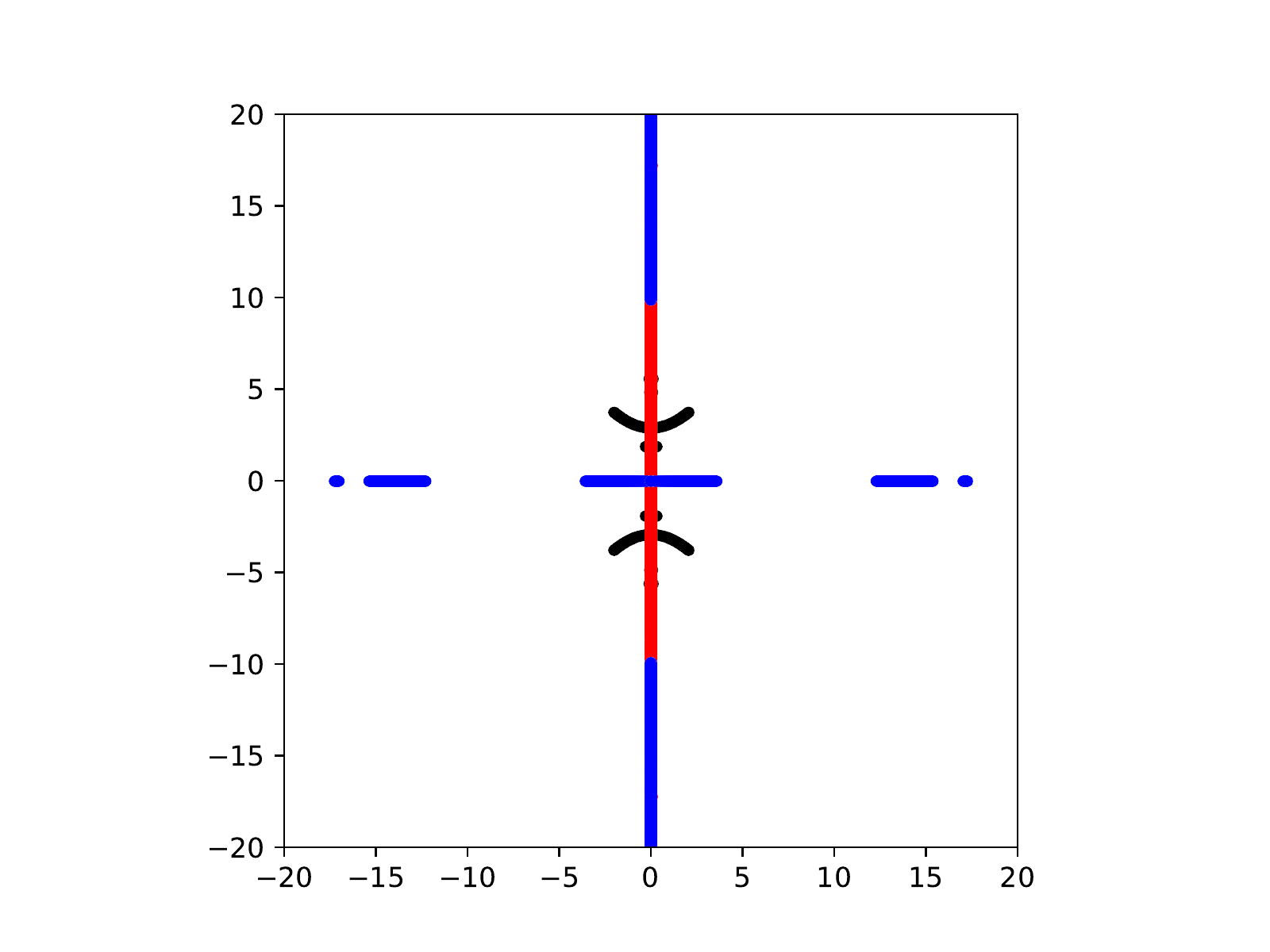}};
       \node (Lax) at (0.9,1.7) {$\sigma_{{\scriptstyle L}}$};
       \node (Stab) at (3.3,1.7)
       {$\sigma_{\scriptstyle{\cL}_{\scriptscriptstyle{\text{dNLS}}}}$};
       \draw [->] (Lax) to [out = 30, in=155] (Stab);
       \node at (2.0, 2.5) {$\Omega$};
    \end{tikzpicture}
      \caption{\label{fig:defocdNLS}}
    \end{subfigure}
    \caption{The real \vs imaginary part of the Lax and stability spectrum
(left and right respectively) for a periodic solution of (a) the focusing dNLS
equation and (b) the defocusing dNLS equation.  \label{fig:NLSspectrum} The Lax
spectrum is computed numerically using \cite{FFHM} and the stability spectrum
is the image under $\Omega$. The Lax spectrum on the real and imaginary axes
are colored blue and red respectively. Their image under the map $\Omega$ is
colored in the stability spectrum accordingly. The Lax spectrum off the real
and imaginary axes and its image under $\Omega$ is black.
    }
  \end{figure}
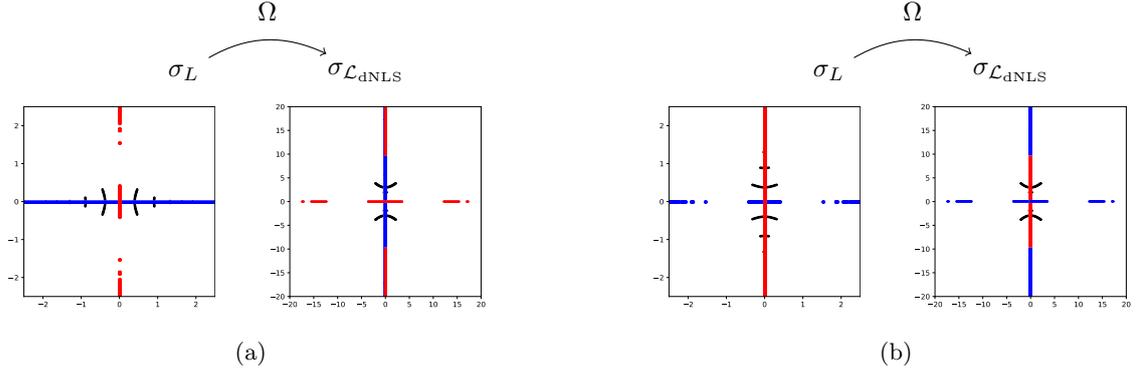

\subsection{Vector and matrix nonlinear Schr\"{o}dinger equations}
The {\em Manakov system} or two-component Vector NLS (VNLS) equation is given by
\begin{align}
  \label{eqn:VNLS}
\begin{split}
  i \pd{q_1}{t} + \pdd[2]{q_1}{x} + 2(\abs{q_1}^2 + \abs{q_2}^2)q_1 &=0,\\
  i \pd{q_2}{t} + \pdd[2]{q_2}{x} + 2(\abs{q_1}^2 + \abs{q_2}^2)q_2 &= 0,
\end{split}
\end{align}
where $q_1$ and $q_2$ are complex-valued functions.  The system
\eqref{eqn:VNLS} was shown to be integrable in \cite{manakovVNLS}. Its
finite-genus solutions (including its elliptic solutions) were explicitly
constructed in \cite{ellipticVNLS}.  Its Lax pair is
\begin{align}\label{eqn:VNLSLaxPair}
  v_x = \begin{pmatrix}
    \alpha & \beta^\intercal \\
    \gamma & \rho
  \end{pmatrix}v = Xv, \qquad v_t = \begin{pmatrix}
    A & B^\intercal \\ C & D
  \end{pmatrix}v = Tv,
\end{align}
with
\begin{equation}
\begin{aligned}
  &\alpha = -i\zeta,& &\beta = q, & &\gamma = -q\cc,& & \rho =
  i\zeta I_2,\\
  &A = -2i \zeta^2 + i q^\intercal q\cc, & &B = 2\zeta q + i q_x, & &C =
  -2\zeta q\cc + i q\cc_x,& & D = 2i\zeta^2I_2 - iq\cc q^\intercal,&
\end{aligned}
\end{equation}
where $q = (q_1,q_2)^\intercal$ and where $I_n$ is the $n\times n$ identity
matrix.  This example does not fit the results found in Section \ref{sec:AKNS}
or Section \ref{sec:generalization}. However, we show that similar results are
found for this system: $\cQ = \{\zeta \in \R: \Omega(\zeta) \in i\R\} \subset
\sigma_L$ and $\Omega(\cQ) \subset i\R\cap \sigma_{\cL}$, as has been
established for other examples.

The compatibility conditions are
\begin{subequations}
\begin{align}
  A_x &= \beta^\intercal C - B^\intercal\gamma,\\
  B_x &= 2\alpha B^\intercal - 2A\beta^\intercal, \\
  C_x &= 2\rho C + 2\gamma A,\\
  D_x &= \gamma B^\intercal - C \beta^\intercal.\label{eqn:compatibilityVNLS}
\end{align}
\end{subequations}
As before, $\Omega$ is found by separation of variables and satisfies
\begin{align}
  \label{eqn:VNLSEigs}
  \begin{pmatrix}
    A-\Omega & B^\intercal \\
    C & D-\Omega I_2
  \end{pmatrix}
  \begin{pmatrix}
    \phi_1 \\ \phi_2
  \end{pmatrix} &= 0,
\end{align}
for nontrivial eigenvectors $\phi = (\phi_1,\phi_2)^\intercal$. Note that $\Omega$ does
not have the form \eqref{eqn:Omega}. Instead, $\Omega$ satisfies
\begin{align}
  0 &= \det\begin{pmatrix}
    A-\Omega & B^\intercal \\ C & D-\Omega I_2
  \end{pmatrix} = \begin{cases}
    (A-\Omega)\det\lt( (D-\Omega I_2) - CB^\intercal/(A-\Omega)\rt), & \Omega \notin
    \sigma(A),\\
    \det(D-\Omega I_2)\lt( (A-\Omega) - B^\intercal(D-\Omega I_2)^{-1}C\rt), & \Omega
    \notin\sigma(D),
  \end{cases}\label{eqn:VNLSOmega}
\end{align}
where $\sigma(L)$ represents the spectrum of $L$. As usual, $\Omega(\zeta)$
defines a Riemann surface.  In the genus-one case~\cite{ellipticVNLS}, it is
represented by
\begin{subequations}
\label{eqn:VNLSRS}
\begin{align}
  f(\zeta,\Omega) &= (\Omega + 2i\zeta^2)(\Omega - 2i\zeta^2)^2 + (2 \lam_2 \zeta + \lam_3)(\Omega
  - 2i\zeta^2) + \mu_0 = 0 ,\label{eqn:VNLSRS}\\
  \lam_2&= -i(q^\intercal \bar{q}_x - q_x^\intercal\bar q),\\
  \lam_3 &= q_x^\intercal\bar q_x + (q^\intercal\bar q)^2,\\
  \mu_0 &= i\abs{q_{1,x} q_2 - q_{2,x}q_1}^2.
\end{align}
\end{subequations}
Since
\begin{align}
  \det(D-\Omega I_2) &= (\Omega - 2i\zeta^2)(A + \Omega),
\end{align}
we use the second expression in \eqref{eqn:VNLSOmega} only when $\Omega =
2i\zeta^2$ or $A+\Omega = 0$. But $\Omega = 2i\zeta^2$ satisfies
\eqref{eqn:VNLSRS} only if $q_2$ and $q_1$ are proportional, and $\Omega = A$
satisfies \eqref{eqn:VNLSRS} only if $\abs{q_1}^2 + \abs{q_2}^2$ is constant.
In the first case, \eqref{eqn:VNLS} reduces to two uncoupled NLS equations, for
which the spectrum is known. In the second case \eqref{eqn:VNLS} reduces to two
uncoupled linear Schr\"{o}dinger equations, for which the spectrum is known.
Therefore, we assume that $D-\Omega I_2$ is invertible.

The eigenfunctions of \eqref{eqn:VNLSLaxPair} are
\begin{align}
  \label{eqn:VNLSEigenFunctions}
  v(x,t) &= e^{\Omega t} y_1(x) \begin{pmatrix}
    a \\ -(D-\Omega I_2)^{-1}Ca
  \end{pmatrix},
\end{align}
where $a\in\C$ is an arbitrary scalar. The scalar function $y_1(x)$ is
determined by substitution in the $x$~equation~\eqref{eqn:VNLSLaxPair}: 
\begin{align}
  y_1' = \lt( \alpha - \beta^\intercal (D-\Omega I_2)^{-1}C\rt)y_1,\label{eqn:VNLSGood}
\end{align}
so that
\begin{align}
  y_1 &= \hat y_1 \exp\lt(\int \lt(\alpha - \beta^\intercal(D-\Omega
  I_2)^{-1}C\rt)~\d x\rt),
\end{align}
where $\hat y_1$ is a constant.  Thus, $\zeta \in \sigma_L$ provided that
\begin{align}
  \left|\RE \int \lt(\alpha - \beta^\intercal(D-\Omega I_2)^{-1}C\rt)~\d x\right|<\infty\qquad
\end{align}
for all $x\in\overline{\R}$. For periodic potentials and $\zeta\in \mathbb{R}$, this
becomes
\begin{align}
  \RE \Langle \beta^\intercal(D-\Omega I_2)^{-1}C\Rangle = 0.
\end{align}
For $\zeta \in \R$,
\begin{align}
  &A^\dagger = -A,& &C^\dagger = -B^\intercal,& & D^\dagger = -D,& &(\beta^\intercal)^\dagger = -\gamma,&
\end{align}
where $F^\dagger = (F\cc)^\intercal$ is the conjugate transpose of $F$. It
follows that $T$ defined by \eqref{eqn:VNLSLaxPair} is skew-adjoint and $\Omega \in
i\R$. Further,
\begin{align}
&(D-\Omega I_2)^\dagger
  = -(D-\Omega I_2).
\end{align}
It follows that
\begin{align}
\begin{split}
  \RE \beta^\intercal(D-\Omega I_2)^{-1}C &= \frac12 \lt[ \beta^\intercal(D-\Omega I_2)^{-1}C +
  C^\dagger ((D-\Omega I_2)^{-1})^\dagger (\beta^\intercal)^\dagger\rt]\\
  &= \frac12 \lt[\Tr \lt( \beta^\intercal(D-\Omega I_2)^{-1}  C -
    B^\intercal (D-\Omega
  I_2)^{-1} \gamma\rt)\rt]\\
  &= \frac12 \lt[\Tr \lt( (D-\Omega I_2)^{-1} C \beta^\intercal - (D-\Omega
  I_2)^{-1}\gamma B^\intercal\rt)\rt]\\
  &= \frac12 \lt[\Tr \lt( (D-\Omega I_2)^{-1}\lt( C\beta^\intercal - \gamma
  B^\intercal\rt)\rt)\rt]\\
  &= \frac12 \lt[ \Tr \lt( (D-\Omega I_2)^{-1} (-D_x)\rt)\rt]\\
  &= -\frac12  \frac{\partial_x \det(D-\Omega I_2)}{\det(D-\Omega I_2)}\\
  &= -\frac12 \partial_x \log \det(D-\Omega I_2),
\end{split}
\end{align}
so that
\begin{align}
  \RE\Langle \beta^\intercal(D-\Omega I_2)^{-1}C\Rangle = 0,
\end{align}
and $\Omega(\zeta) \in i\R$ for $\zeta \in \R$. This can be verified using the
method described in \cite{FFHM}.

The work above can be generalized to the Matrix NLS (MNLS) equation,
\begin{align}
  i U_t + U_{xx} -2\kappa U U\cc U &= 0,\label{eqn:MatrixNLS}
\end{align}
where $U$ is an $\ell_1\times \ell_2$ matrix for $\ell_1,\ell_2\in \N$, and
$\kappa = -1$ and $\kappa = 1$ correspond to the focusing and defocusing cases
respectively. The Lax pair for the MNLS equation is given by
\begin{subequations}
\begin{align}
  \Psi_x &= \begin{pmatrix}
    -i\zeta I_{\ell_1} & Q\\
    R & i\zeta I_{\ell_2}
  \end{pmatrix} \Psi = X\Psi,\\
  \Psi_t &= \begin{pmatrix} A & B \\ C & D\end{pmatrix}\Phi = T\Psi,\\
  A = -2i \zeta^2 I_{\ell_1} - i QR, \qquad B &= 2\zeta Q + i Q_x, \qquad C =
  2\zeta R - iR_x, \qquad D = 2i\zeta^2 I_{\ell_2}+iRQ.
\end{align}
\end{subequations}
Here, $Q$ and $R$ are $\ell_1\times \ell_2$ and $\ell_2\times \ell_1$
complex-valued matrices respectively. The $x$ equation may be written as a
spectral problem,
\begin{align}
  \zeta \Psi &= \begin{pmatrix}
    iI_{\ell_1} \partial_x & -iQ\\
    iR & -i I_{\ell_2}\partial_x
  \end{pmatrix}
  \Psi = L\Psi.
\end{align}
$L$ is self adjoint if $R\cc = Q$, hence $\sigma_L \subset \R$ if $R\cc = Q$.
The compatibility conditions are the same as \eqref{eqn:compatibilityVNLS}. The
steps above apply in a straightforward but cumbersome manner to establish that
$\cQ = \{\zeta \in \R: \Omega(\zeta)\in i\R\}\subset \sigma_L$ and
$\Omega(\cQ)\subset i\R\cap \sigma_{\cL}.$ Details are omitted here for brevity.

\section{Conclusion}
The stability spectrum and the Lax spectrum for solutions of many integrable
equations on the whole line have been characterized for some time. The same level of
understanding for the periodic problem does not exist. One reason the whole
line problem is more straightforward to study is the ability to do spatial
asymptotics to find the essential spectrum that contains the unbounded
components of the spectrum. In this work, \te{we have given an asymptotic
characterization of all unbounded components of the Lax spectrum for a number
of integrable equations, using \eqref{eqn:newLaxSpectrumConditions} (see Remark
\ref{remark:asymptotics2}). }

We provided two theorems (Theorems~\ref{thm:AKNSThm} and \ref{thm:nonAKNSThm})
with easily verifiable assumptions that establish that real Lax spectra
corresponds to stable modes of the linearization for a number of equations in
and not in the AKNS hierarchy. We applied the theorems to a number of examples.
The methods described in this paper can be applied to other equations not
mentioned in Sections \ref{sec:AKNSExamples} and \ref{sec:examples}. Some other
examples include: the KdV equation \cite{bottman2009kdv}, Hirota's equation (or
the mixed generalized NLS-generalized mKdV equation) \cite{MR0001169,
kodama1987nonlinear}, the Modified Vector dNLS equation \cite{DMVEquation}, the
Massive Thirring Model \cite{kaupMTM}, the O$_4$ nonlinear $\sigma$-model
\cite{lund1977example}, the complex reverse space-time nonlocal mKdV equation
\cite{ablowitz-nonlocalHierarchy}, and the reverse space-time nonlocal
generalized sine-Gordon equations \cite{ablowitz-nonlocalHierarchy}.

\begin{appendices}

\section{The Floquet discriminant \label{sec:FloquetDiscriminant}}
A common tool for characterizing the Lax spectrum for periodic potentials is the
Floquet discriminant \cite{ablowitz1996computational, calini2011squared,
MR1123280, leeThesis}. The Floquet discriminant is typically approximated
numerically since the eigenfunctions of the $x$ equation are unknown for generic
potentials. In our framework, we have explicit expressions for the eigenfunctions
\eqref{eqn:theAKNSEigenfunctions}.  Since $\Omega(\zeta)$ is defined by its square,
\eqref{eqn:Omega} defines two different values of $\Omega$ for every value of
$\zeta$ for which $\Omega(\zeta)\neq 0$. Hence
\eqref{eqn:theAKNSEigenfunctions} defines two linearly independent solutions of
\eqref{eqn:isospectralLaxPair} except for when $\Omega(\zeta)=0$. When
$\Omega(\zeta)=0$, only one solution is generated by
\eqref{eqn:theAKNSEigenfunctions} and a second solution is found using the
method of reduction of order. The solution found by reduction of order is
algebraically unbounded so it is not an eigenfunction. For $\Omega(\zeta)\neq
0$, the two eigenfunctions of \eqref{eqn:isospectralLaxPair} are
\begin{align}
  \phi_\pm(x,t) &= e^{\pm \Omega t} y_{\pm}(x) \begin{pmatrix} - \hat B(x) \\
    \hat A(x) - \Omega_{\pm}\end{pmatrix}.
\end{align}
We use one choice of \eqref{eqn:theAKNSEigenfunctions} and one choice of $y_1$.
The following computations proceed similarly for the other choices. A
fundamental matrix solution (FMS) of the $x$-equation of
\eqref{eqn:isospectralLaxPair} is 
\begin{align}
  M(x) &= \begin{pmatrix}
    - \hat B(x) y_+(x) & - \hat B(x)y_-(x)\\
    (\hat A(x)-\Omega_+)y_+(x) & (\hat A(x)-\Omega_-)y_-(x)
  \end{pmatrix},
\end{align}
where dependence on $t$ has been omitted. The FMS normalized to the identity
is given by
\begin{align}
  \tilde M(x; x_0) &= M^{-1}(x_0)M(x).
\end{align}
To simplify notation, we define
\begin{align}
  I_{\pm}(x;\zeta) &= -\int \lt(\hat \alpha + \frac{\hat\beta \hat C}{\hat
  A-\Omega_\pm}\rt)~\d x.
\end{align}
In this section we use Assumption~\ref{assumption:alphaPeriodic} instead of
assuming that $\alpha \in i\R$. 
\begin{assumption}\label{assumption:alphaPeriodic}
  $\hat \alpha$ and $\hat\beta C/(\hat A-\Omega)$ are periodic in $x$ with
  the same period $P$ as the solution.
\end{assumption}
\noindent Under Assumption~\ref{assumption:alphaPeriodic}, each of the
integrands in \eqref{eqn:newLaxSpectrumConditions} is $P$-periodic, and we may
use any of the representations for $I$. It follows that
\begin{align}
  I_{\pm}(x+P;\zeta) &= I_{\pm}(x;\zeta) + I_{\pm}(P;\zeta).
\end{align}
Then
\begin{align}
  y_{\pm}(x+P) &= y_{\pm}(x) e^{I_{\pm}(x;\zeta)}e^{I_{\pm}(P;\zeta)} =
  y_{\pm}(x) \Gamma_{\pm}(P),
\end{align}
and
\begin{align}
  \tilde M(x+P; x_0) &= M^{-1}(x_0)M(x+P) = M^{-1}(x_0)M(x)\Gamma(P) = \tilde
  M(x;x_0) \Gamma(P),
\end{align}
where
\begin{align}
  \Gamma(P) &= \begin{pmatrix} \Gamma_+(P) & 0 \\ 0 & \Gamma_-(P)\end{pmatrix}
\end{align}
is the transfer matrix. In order for solutions to be bounded in space, it must
be that the eigenvalues of the transfer matrix have unit modulus. Thus,
\begin{align}
  \RE\lt( I_{\pm}(P;\zeta)\rt) = 0.
\end{align}
If Assumption~\ref{assumption:alphaPeriodic} holds, this is equivalent to
\eqref{eqn:newLaxSpectrumConditions}. The Floquet discriminant is defined by
\begin{align}
  \Delta(\zeta) &= \tr(\Gamma(P)) = \Gamma_+(P) + \Gamma_-(P),
\end{align}
and
\begin{align}
  \sigma_L = \lt\{ \zeta \in \C : \IM(\Delta(\zeta)) = 0 \quad \text{and} \quad
  \abs{\Delta(\zeta)}\leq 2\rt\}.\label{eqn:LaxSpectrumFloquet}
\end{align}
Both definition \eqref{eqn:LaxSpectrumFloquet} and
\eqref{eqn:newLaxSpectrumConditions} require numerical computation or the use of
special functions. We prefer working with \eqref{eqn:newLaxSpectrumConditions}
directly, but we present the Floquet discriminant because of its popularity.

\end{appendices}

{
  \footnotesize
  \bibliographystyle{siam}
  \bibliography{mybib}
}

\end{document}